\pdfoutput=1
\documentclass[final,3p,times]{elsarticle}

\usepackage{amsmath,amssymb,amsthm,mathtools}
\usepackage{bm}
\usepackage{booktabs}
\usepackage{graphicx}
\usepackage{subcaption}
\usepackage{placeins}
\usepackage{multirow}
\usepackage{array}
\usepackage[ruled,vlined,linesnumbered]{algorithm2e}
\usepackage{enumitem}
\usepackage{tikz}
\usetikzlibrary{arrows,positioning,fit,backgrounds,calc}
\usepackage{hyperref}

\newtheorem{theorem}{Theorem}
\newtheorem{lemma}{Lemma}

\newtheorem{proposition}{Proposition}
\newtheorem{remark}{Remark}

\newtheorem{assumption}{Assumption}

\newcommand{\R}{\mathbb{R}}

\newcommand{\cA}{\mathcal{A}}

\newcommand{\cB}{\mathcal{B}}
\newcommand{\cN}{\mathcal{N}}

\newcommand{\cU}{\mathcal{U}}
\newcommand{\cR}{\mathcal{R}}
\newcommand{\cE}{\mathcal{E}}

\newcommand{\bx}{\mathbf{x}}
\newcommand{\bz}{\mathbf{z}}
\newcommand{\bw}{\mathbf{w}}
\newcommand{\bc}{\mathbf{c}}
\newcommand{\bA}{\mathbf{A}}
\newcommand{\bb}{\mathbf{b}}

\newcommand{\bzero}{\mathbf{0}}
\newcommand{\lr}{\mathrm{L\text{-}RFM}}
\newcommand{\TableFont}{\footnotesize}

\newcommand{\FigDiagnosticWidth}{0.82\linewidth}
\newcommand{\FigFieldWidth}{0.88\linewidth}

\newcommand{\FigSchematicWidth}{0.92\textwidth}

\newcommand{\maybeincludegraphics}[2][]{%
  \IfFileExists{#2}{\includegraphics[#1]{#2}}{%
    \fbox{\parbox{0.9\linewidth}{Missing figure: \texttt{\detokenize{#2}}}}%
  }%
}

\journal{}

\begin{document}

\begin{frontmatter}

\title{Liquid Random Feature Methods for Time-Dependent Partial Differential Equations}

\author[xidian]{Jiale Linghu}
\author[nus]{Yangshuai Wang\corref{cor1}}
\ead{yswang@nus.edu.sg}
\cortext[cor1]{Corresponding author.}
\address[xidian]{School of Mathematics and Statistics, Xidian University, Xi'an, China}
\address[nus]{Department of Mathematics, National University of Singapore, Singapore}

\begin{abstract}
A central challenge in mesh-free space--time approximation for time-dependent partial differential equations is to represent evolving temporal scales while keeping residual minimization computationally tractable.
Random feature methods simplify this algebraic problem by freezing nonlinear trial functions and fitting only a linear readout, but standard static space--time activations provide no explicit relaxation-scale mechanism, making temporal-scale resolution a finite-dimensional bottleneck in stiff, dispersive, or multi-scale regimes.
We introduce liquid random feature methods (L-RFM), which replace static temporal activations by closed-form liquid time-constant responses with sampled relaxation scales.
The resulting frozen features form temporally structured local or global trial spaces with analytic space--time derivatives for residual least-squares assembly.
A density theorem proves density of the deterministic trial spaces in the continuous space--time function class, and a temporal-rank calculation clarifies the role of sampled relaxation scales. Ablation and finite-feature tests identify the liquid temporal response as the primary source of the observed accuracy improvement.
Across stiff reaction--diffusion, nonlinear transport, dispersive, complex-valued, and multidimensional benchmarks, L-RFM improves finite-feature accuracy in regimes where temporal-scale representation controls the approximation.
By embedding relaxation scales directly into frozen trial functions, L-RFM provides a route to high-accuracy continuous space--time surrogates for evolutionary PDEs while preserving the simplicity of linear least-squares solvers.
\end{abstract}

\begin{keyword}
random feature method \sep liquid time-constant networks \sep time-dependent PDEs \sep least squares \sep partition of unity
\end{keyword}

\end{frontmatter}

\section{Introduction}
\label{sec:introduction}

Time-dependent partial differential equations (PDEs) are central models in computational science and engineering, from diffusion and transport to waves and fluid motion \cite{leveque2007finite}.
Their numerical solution requires spatial structure and evolving temporal scales to be represented together, a requirement shared by time-marching and space--time formulations \cite{leveque2007finite,hughes1988spacetime}.
Problem-tailored solvers for highly oscillatory quantum dynamics, nonlocal Bose--Einstein condensates, singular Schr{\"o}dinger models, and geometric interface evolution further show that temporal-scale representation and structure preservation are often decisive for accuracy and stability \cite{bao2016diracMTI,bao2016dipolarNUFFT,bao2019logSchrodinger,bao2017dewettingPFEM}.
Mesh-free space--time collocation and residual-based neural or random-feature solvers extend this viewpoint by constructing continuous surrogates and off-grid derivatives without prescribing a space--time mesh \cite{raissi2019pinn,karniadakis2021review,dwivedi2020pielm,chen2023strfm}.
Their practical accuracy depends on the ability of a finite trial space to carry the temporal scales of the evolution while keeping residual minimization tractable.
This temporal representation problem is especially acute in stiff, dispersive, or multi-scale regimes, where improving sampling or linear algebra alone cannot compensate for trial functions that do not contain the relevant time scales.

In residual least-squares space--time methods, the unknown field is approximated directly on the space--time domain, and the governing equation together with the initial and boundary conditions is enforced at sampled collocation points.
This converts an evolutionary solve into a finite-dimensional approximation problem, with the temporal content of the trial functions playing a decisive role.
If those functions contain only static space--time activations, fast transients, slow relaxation, and interacting temporal scales must be represented indirectly through the sampled activation geometry.

Existing residual-based mesh-free solvers make different compromises between approximation flexibility and algebraic simplicity.
Physics-informed neural networks (PINNs) enforce PDE residuals by training all network parameters \cite{raissi2019pinn,karniadakis2021review}, but the resulting nonconvex optimization is often sensitive to architecture, sampling, and loss weighting \cite{wang2022ntk,wu2023comprehensive}.
Operator-learning methods, including DeepONet \cite{lu2021deeponet} and the Fourier neural operator \cite{li2021fno}, learn maps from problem data to solution fields and are therefore complementary to one-shot forward solvers.
Random feature methods (RFMs) offer a more direct route for such solves: nonlinear features are sampled once and frozen, and only the linear readout is fitted by least squares \cite{rahimi2007random,chen2022rfm,dwivedi2020pielm,huang2006elm}.
This fixed-feature structure preserves a simple linear algebraic core, but for evolutionary PDEs the key unresolved issue is how the frozen trial functions should encode temporal response.

The closest least-squares predecessors are PIELM \cite{dwivedi2020pielm}, XTFC \cite{schiassi2021xtfc}, the PoU random feature method of Chen, Chi, E, and Yang \cite{chen2022rfm}, the space--time random feature method (ST-RFM) \cite{chen2023strfm}, and related local or domain-decomposed ELM/RFM solvers \cite{dong2021localelm,dong2022hyperparameter,ni2023hcelm,wang2024elmhighdim,vanbeek2026featurefilter,anderson2026elmfbpinn}.
Recent RFM developments have also addressed interface, discontinuity, and complex-domain elliptic problems \cite{chi2024interface,sun2025twolevel,song2026discontinuity}, while multilevel domain-decomposition and hybrid-learning formulations have improved the scalability and reliability of physics-informed or fixed-basis architectures \cite{dolean2024multilevel,wu2026hybrid}.
These methods show that fixed hidden features, PoU localization, and domain-decomposition structure can produce efficient residual collocation solvers, but their temporal components remain static space--time activations or ridges.
L-RFM keeps the same least-squares philosophy and, in the local variant, the same PoU localization principle, while replacing the static temporal response by a closed-form liquid relaxation response.

The liquid time-constant (LTC) model motivates this mechanism through hidden dynamics with input-dependent relaxation \cite{hasani2021ltc}.
Physics-informed liquid networks have also been used by optimizing liquid-network parameters or hidden states against PDE residuals \cite{sun2024piln}, alongside broader PINN and operator-learning variants for nonlinear and heterogeneous PDE mappings \cite{huang2022hompinns,zheng2024hompinns,zhang2024d2no}.
L-RFM uses the liquid response differently.
The scalar dynamics are solved in closed form, the parameters are sampled once, and the resulting functions are frozen before residual assembly.
Thus the method remains a random-feature least-squares solver rather than a trained liquid network.
It is also related to reservoir computing, echo state networks, and random dynamical features \cite{jaeger2001esn,maass2002lsm,lukosevicius2009reservoir,racca2021apiesn,piatti2026rcde}, but the fitted readout is determined by space--time PDE residual collocation rather than by sequence regression.

We address this temporal representation gap with liquid random feature methods (L-RFM).
The central idea is to replace static temporal activations by closed-form liquid time-constant responses whose relaxation scales are sampled and then frozen.
The method therefore changes where temporal information enters the approximation: not through nonlinear training or a time-marching update, but through the construction of the frozen trial functions themselves.
At the feature level, the temporal response has the form
\begin{equation}
\phi(\bx,t)=s(\bx)+\bigl(h^0(\bx)-s(\bx)\bigr)\exp\!\bigl(-(\tau^{-1}+g(\bx))t\bigr),\quad
s(\bx)=\frac{g(\bx)A}{\tau^{-1}+g(\bx)}.
\label{eq:intro-closed}
\end{equation}
Here $\bx$ and $t$ denote space and time, $\phi$ is a frozen scalar feature, $\tau>0$ is a sampled relaxation time, $A$ is a sampled amplitude, and $g(\bx)$ and $h^0(\bx)$ are sampled spatial profiles.
The quantity $s(\bx)$ is the local steady response, while $\tau^{-1}+g(\bx)$ sets the decay rate toward that response.
Thus the word ``liquid'' refers to a frozen closed-form relaxation response, not to a trained recurrent network in the least-squares solve.
The closed form provides analytic derivatives for PDE residual assembly while keeping the readout linear.
Sampling $\log\tau$ across decades gives the finite trial space an explicit multi-scale relaxation spectrum, rather than relying on static activations to recover temporal scales indirectly.

The same temporal primitive leads to two complementary trial-space realizations.
The localized form multiplies liquid features by PoU weights on spatial patches, whereas the global form uses the same liquid features in a single non-localized trial space.
Because the two realizations share the temporal response, analytic derivative structure, and LS/Picard residual assembly, their comparison separates the effect of temporal relaxation from that of spatial localization.
In both cases, all hidden parameters remain fixed and only least-squares readout problems are solved; nonlinear equations are treated by Picard-linearized residual systems, and long intervals can be handled by block marching.
The analysis has two complementary roles: it proves density of deterministic local and global trial spaces through a separable ridge--exponential subfamily, and it gives a finite-dimensional temporal-rank calculation that explains how sampled relaxation scales expand the temporal representation available to the trial space.
The numerical results emphasize the finite-feature regime that determines practical performance in residual collocation, identifying how liquid temporal response improves accuracy at fixed feature counts and when spatial localization is needed.
This positions L-RFM as a temporally structured fixed-feature framework for regimes in which temporal-scale representation is central to the approximation.

The contribution of this work is fourfold.
First, we formulate L-RFM as a fixed-feature least-squares collocation framework in which the temporal trial functions contain sampled liquid relaxation scales.
Second, we develop local and global trial spaces based on the same liquid primitive, so that the effects of temporal relaxation and spatial localization can be examined within a common formulation.
Third, we derive closed-form space--time derivatives and residual systems for linear, nonlinear, complex-valued, multidimensional, and block-marched problems, without differentiating through a numerical ODE solver.
Finally, we combine density and temporal-rank results for the deterministic trial spaces with numerical studies showing when the liquid temporal response improves finite-feature accuracy for stiff, dispersive, nonlinear, and multidimensional evolutionary PDEs.

This paper is organized as follows.
Section~\ref{sec:method} defines the liquid feature primitive, the localized and global trial spaces, the analytic derivative formulas, and the least-squares residual assembly used for linear, nonlinear, complex-valued, and block-marched solves.
Section~\ref{sec:theory} proves density of the deterministic L-RFM trial spaces and gives a finite-dimensional temporal-rank calculation that clarifies the role of sampled relaxation scales.
Section~\ref{sec:experiments} verifies the implementation and evaluates the method on representative time-dependent PDEs, including local/global comparisons, ablations, conditioning measurements, multidimensional checks, long-time marching, and timing.
Section~\ref{sec:discussion} summarizes the regimes in which L-RFM is most useful, the current theoretical and computational limitations, and the main directions for further development.

\section{Liquid Random Feature Methods}
\label{sec:method}

This section defines the L-RFM discretization used throughout the paper.
The construction has three components: a closed-form liquid feature primitive, a spatial support rule that yields local or global trial spaces, and a residual least-squares assembly for linear, Picard-linearized nonlinear, and block-marched solves.
We consider the initial-boundary-value problem
\begin{equation}
\begin{aligned}
\partial_t u(\bx,t)+\cN[u](\bx,t)&=f(\bx,t), &&(\bx,t)\in\Omega\times(0,T],\\
u(\bx,0)&=u_0(\bx), &&\bx\in\overline{\Omega},\\
\cB[u](\bx,t)&=g_{\partial}(\bx,t), &&(\bx,t)\in\partial\Omega\times(0,T],
\end{aligned}
\label{eq:method-ibvp}
\end{equation}
where $\Omega\subset\R^d$ is bounded and Lipschitz, $\cN$ is a possibly nonlinear spatial differential operator, and $\cB$ is a linear boundary operator.
The presentation focuses on the trial-space construction and the resulting residual-collocation algebraic system for a well-posed problem of this form.
All definitions are written for general spatial dimension $d$, with the same assembly used in one, two, and three dimensions.
Figure~\ref{fig:lrfm-schematic} gives the complete workflow before the individual components are defined.
The construction starts from the IBVP data, samples frozen liquid features with prescribed relaxation scales, assigns local or global spatial support, evaluates analytic derivatives, and solves the resulting weighted least-squares readout problem.

\begin{center}
\centering
\maybeincludegraphics[width=\FigSchematicWidth]{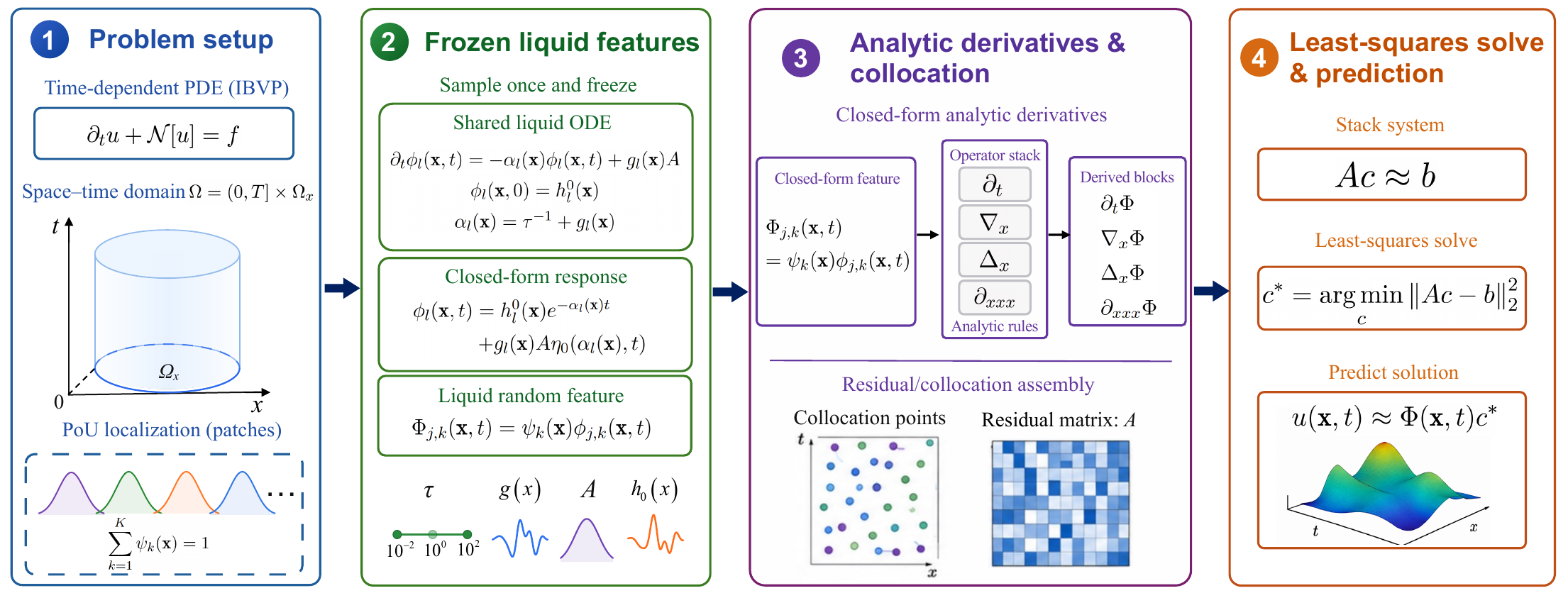}
\captionof{figure}{Overview of the L-RFM construction and residual least-squares solve. The method forms local or global space--time trial functions from frozen liquid random features, evaluates closed-form responses and analytic derivatives, assembles the weighted collocation residual, and solves the linear least-squares readout problem.}
\label{fig:lrfm-schematic}
\end{center}

\subsection{Shared liquid primitive and local/global trial spaces}
\label{subsec:trial-space}

We first specify the shared feature family and then define the local and global trial spaces.
Stable derivative formulas, residual least-squares assembly, Picard linearization, and block marching are developed in Sections~\ref{subsec:closed-form}--\ref{subsec:algorithm}.

\subsubsection{Shared closed-form liquid primitive}
We separate the construction into a temporal primitive and a spatial support rule.
The temporal primitive is shared by all L-RFM variants; locality enters through the support assigned to the final trial functions.
Let
\[
\theta=(\tau,A,\bw,b,\bw^0,b^0)
\]
be a sampled parameter tuple.
On a spatial coordinate chart $\zeta(\bx)$, which is either a local patch coordinate $\bz_k(\bx)$ or a global affine normalization $\bz(\bx)$, define
\begin{equation}
g_\theta(\bx)=\tanh\bigl(\bw^\top\zeta(\bx)+b\bigr),\qquad
h^0_\theta(\bx)=\tanh\bigl((\bw^0)^\top\zeta(\bx)+b^0\bigr),
\label{eq:gh0}
\end{equation}
and
\[
\alpha_\theta(\bx)=\tau^{-1}+g_\theta(\bx).
\]
The sampled rate $\alpha_\theta$ is allowed to be positive, small, or sign-changing on the finite computational window; Section~\ref{subsec:closed-form} gives the bounded-window estimate and removable-limit evaluation used near $\alpha=0$.
The feature $\phi_\theta$ is the solution of the scalar frozen liquid ODE
\begin{equation}
\partial_t\phi_\theta(\bx,t)=-\alpha_\theta(\bx)\phi_\theta(\bx,t)+g_\theta(\bx)A,\qquad
\phi_\theta(\bx,0)=h^0_\theta(\bx).
\label{eq:ltc-feature}
\end{equation}
For stable numerical evaluation, we use the Duhamel form
\begin{equation}
\begin{aligned}
\phi_\theta(\bx,t)
&=h^0_\theta(\bx)e^{-\alpha_\theta(\bx)t}
+g_\theta(\bx)A\,\eta_0(\alpha_\theta(\bx),t),\\
\eta_0(\alpha,t)&=\frac{1-e^{-\alpha t}}{\alpha},\qquad \eta_0(0,t)=t,
\end{aligned}
\label{eq:duhamel-feature}
\end{equation}
This is the evaluation form used in assembly; Section~\ref{subsec:closed-form} records its equivalent steady-state representation.
The same primitive, sampling law, analytic derivative chain, residual assembly, and least-squares readout are used by both L-RFM-Local and L-RFM-Global; the two variants differ in the spatial support assigned to this primitive.

The frozen parameters are drawn from
\begin{equation}
\begin{gathered}
\log_{10}\tau\sim\cU\bigl[\log_{10}\tau_{\min},\log_{10}\tau_{\max}\bigr],\qquad
\bw,\bw^0\sim\cU[-R_w,R_w]^d,\\
b,b^0\sim\cU[-R_b,R_b],\qquad A\sim\cU[-1,1],
\end{gathered}
\label{eq:sampling}
\end{equation}
with $\tau_{\min}=c_{\min}T_{\rm win}$ and $\tau_{\max}=c_{\max}T_{\rm win}$, where $T_{\rm win}=T$ for a whole-window solve and $T_{\rm win}$ is the current block length in block-marching runs.
The constants $(c_{\min},c_{\max},R_w,R_b)$ are fixed trial-space design hyperparameters that remain unchanged during the least-squares solve and are specified by the benchmark configurations.
The log-uniform law for $\tau$ is the design choice that places relaxation directions across multiple temporal scales.

For either variant, let $\{\Phi_q\}_{q=1}^P$ denote the final set of trial functions after the spatial support convention has been chosen.
The numerical solution has the common linear readout form
\begin{equation}
u_{\lr}(\bx,t)=\sum_{q=1}^{P}c_q\Phi_q(\bx,t).
\label{eq:unified-ansatz}
\end{equation}
Thus L-RFM-Local and L-RFM-Global differ in the spatial support convention used to define $\Phi_q$.
Here, ``local'' refers to $\Phi_q=\psi_k\phi_{i,k}$ with compact patch support and $P=KM$, whereas ``global'' refers to $\Phi_q=\phi_i$ with full spatial support and one global readout.

\subsubsection{L-RFM-Local: PoU-localized liquid trial space}
For the local variant, the domain is covered by $K$ overlapping subdomains $\{\Omega_k\}_{k=1}^K$ with centers $\bx_k^c$.
Each patch has a diagonal scaling
\[
D_k=\mathrm{diag}(r_{k,1},\dots,r_{k,d}),
\]
and local coordinate $\bz_k(\bx)=D_k^{-1}(\bx-\bx_k^c)$.
A smooth partition of unity (PoU) $\{\psi_k\}\subset C^\infty(\overline{\Omega})$ satisfies $\psi_k\ge 0$, $\mathrm{supp}\,\psi_k\subset\Omega_k$, and $\sum_k\psi_k\equiv 1$ on $\overline{\Omega}$ \cite{chen2022rfm}.
The PoU bumps in this paper are built from the standard $C^\infty$ compact-support bump
\begin{equation}
b(r)=
\begin{dcases}
\exp\!\Bigl(-\frac{1}{1-r^2}\Bigr), & |r|<1,\\
0, & |r|\ge 1,
\end{dcases}
\label{eq:pou-bump}
\end{equation}
which is smooth across $|r|=1$ together with all its derivatives.
For $d$-dimensional tensor-product patches we set
\[
\tilde\psi_k(\bx)=\prod_{j=1}^d b\bigl(z_{k,j}(\bx)\bigr),
\]
with the one-dimensional case understood as the scalar product.
Provided the cover $\{\Omega_k\}$ is chosen so that every $\bx\in\overline{\Omega}$ lies strictly inside at least one $\Omega_k$, the sum $\sum_\ell\tilde\psi_\ell$ is positive on $\overline{\Omega}$, and the normalization $\psi_k=\tilde\psi_k/\sum_\ell\tilde\psi_\ell$ enforces $\sum_k\psi_k\equiv 1$.
Because $b$ and all its derivatives extend by zero across $|r|=1$, each $\psi_k$ is globally $C^\infty(\overline{\Omega})$, and the spatial derivatives of $\psi_k$ needed by the residual operator are obtained from \eqref{eq:pou-bump} by elementary calculus.

On each patch, L-RFM-Local samples $M$ independent parameter tuples $\theta_{i,k}$ from \eqref{eq:sampling} and instantiates the shared primitive with the local coordinate $\zeta(\bx)=\bz_k(\bx)$.
We write $\phi_{i,k}$, $g_{i,k}$, $h^0_{i,k}$, and $\alpha_{i,k}$ for the corresponding local instantiations of \eqref{eq:gh0}--\eqref{eq:ltc-feature}.
The readout dimension is $P=KM$ and the feature index is $q=(k,i)$.
The localized trial functions are
\begin{equation}
\Phi_q(\bx,t)=\Phi^{\mathrm{loc}}_{i,k}(\bx,t)=\psi_k(\bx)\phi_{i,k}(\bx,t),
\end{equation}
so that
\begin{equation}
u_{\mathrm{L\text{-}RFM\text{-}Local}}(\bx,t)
=\sum_{k=1}^K\psi_k(\bx)\sum_{i=1}^M c_{i,k}\,\phi_{i,k}(\bx,t),
\label{eq:lrfm-ansatz}
\end{equation}
with trainable readout vector $\bc=(c_{i,k})\in\R^{KM}$.
The trial functions $\psi_k\phi_{i,k}$ are spatially local because the PoU factors vanish outside their patches.

\subsubsection{L-RFM-Global: global liquid trial space}
For the global variant, the PoU localization is removed and one global affine normalization $\zeta(\bx)=\bz(\bx)$ is used over the whole spatial domain.
We sample $P$ independent copies $\theta_i$ of the same shared liquid primitive, set $q=i$, and define
\[
\Phi_q(\bx,t)=\Phi^{\mathrm{glob}}_i(\bx,t)=\phi_i(\bx,t).
\]
We write $\phi_i$ for this global instantiation below.
Thus
\begin{equation}
u_{\mathrm{L\text{-}RFM\text{-}Global}}(\bx,t)
=\sum_{i=1}^{P}c_i\,\phi_i(\bx,t),
\label{eq:lrfm-global-ansatz}
\end{equation}
with a single global readout vector $\bc\in\R^P$.
The global trial functions have full spatial support and otherwise use the same liquid primitive, parameter law, derivative formulas, and LS readout as the local construction.
Only the readout coefficients are fitted in either variant; the sampled parameters in \eqref{eq:gh0}--\eqref{eq:ltc-feature} remain fixed.
Matched Local/Global comparisons use the same total readout dimension $P$ and therefore isolate the effect of spatial support.

\subsection{Stable closed form and analytic derivatives}
\label{subsec:closed-form}

For fixed $\bx$, the shared primitive \eqref{eq:ltc-feature} is a linear first-order ODE in $t$ with $t$-independent coefficients.
The formulas below are written for a local feature $(i,k)$; the global case is obtained by dropping the patch index and replacing $\bz_k$ with the global coordinate $\bz$.
We evaluate the Duhamel form \eqref{eq:duhamel-feature} because it has a removable limit at $\alpha_{i,k}=0$.
When $\alpha_{i,k}(\bx)\ne0$, the same expression can be rearranged into the steady-state form
\begin{equation}
\phi_{i,k}(\bx,t)=s_{i,k}(\bx)+\bigl(h^0_{i,k}(\bx)-s_{i,k}(\bx)\bigr)\,E_{i,k}(\bx,t),
\label{eq:closed-form}
\end{equation}
where
\[
s_{i,k}=\frac{g_{i,k}A_{i,k}}{\alpha_{i,k}},\qquad
E_{i,k}=\exp\!\bigl(-\alpha_{i,k}(\bx)t\bigr).
\]
At $\alpha_{i,k}=0$, the Duhamel form \eqref{eq:duhamel-feature} gives the finite limit $\phi_{i,k}=h^0_{i,k}+g_{i,k}A_{i,k}t$.
In this representation, $s_{i,k}(\bx)$ is the constant-in-$t$ offset and the $t\to\infty$ asymptote when $\alpha_{i,k}>0$.
The function $h^0_{i,k}(\bx)$ is the initial trace, and $h^0_{i,k}-s_{i,k}$ is the coefficient of the exponential factor $E_{i,k}$ whose rate depends on $\bx$.
The two ridge functions \eqref{eq:gh0} therefore set independent spatial shapes for the asymptote and the initial state, while the sampled $\tau_{i,k}$ controls the relaxation time scale.

Since $|g_{i,k}(\bx)|\le 1$ and $\tau_{i,k}\in[\tau_{\min},\tau_{\max}]$, the rate $\alpha_{i,k}(\bx)$ lies in the bounded interval $[\tau_{\max}^{-1}-1,\,\tau_{\min}^{-1}+1]$.
On each finite solve window this gives $|\alpha_{i,k}(\bx)|\le\tau_{\min}^{-1}+1$ uniformly, so $|E_{i,k}(\bx,t)|\le\exp(|\alpha_{i,k}|T_{\rm win})$ stays bounded.
Features drawn with small $\tau$ act as fast relaxation kernels, whereas those with large $\tau$ and negative $g$ provide bounded finite-window growing components.
The trial space therefore contains both decaying and growing temporal directions, and the SVD-regularized least-squares readout selects the linear combination that fits the residual.

The same removable-limit convention is used for derivatives.
Let
\[
\eta_m(\alpha,t)=\partial_\alpha^m\eta_0(\alpha,t),\qquad m\ge1,
\]
with limits evaluated by Taylor expansion when $|\alpha|$ is small; for example $\eta_0(0,t)=t$, $\eta_1(0,t)=-t^2/2$, and $\eta_2(0,t)=t^3/3$.
Writing $h=h^0_{i,k}$, $g=g_{i,k}$, $\alpha=\alpha_{i,k}$, $E=e^{-\alpha t}$, and using subscripts $j$ and $jj$ for first and second derivatives in $x_j$, differentiating the Duhamel form gives the stable first derivative
\begin{equation}
\begin{aligned}
\partial_{x_j}\phi_{i,k}
= h_jE-h\alpha_j tE
+ A_{i,k}\bigl(g_j\eta_0+g\alpha_j\eta_1\bigr),
\end{aligned}
\label{eq:dxphi}
\end{equation}
where all $\eta_m$ are evaluated at $(\alpha,t)$ and $\alpha_j=g_j$.
Elementary calculus gives $g_j=(1-g^2)w_{i,k,j}/r_{k,j}$ and $h_j=(1-h^2)w^0_{i,k,j}/r_{k,j}$ for a local tensor-product patch, with the corresponding affine scaling used in the global coordinate.
The second derivative in the same direction is
\begin{equation}
\begin{aligned}
\partial_{x_j}^2\phi_{i,k}
=\;&h_{jj}E-2h_j\alpha_j tE
+h\bigl(\alpha_j^2t^2-\alpha_{jj}t\bigr)E\\
&+A_{i,k}\Bigl[
g_{jj}\eta_0+2g_j\alpha_j\eta_1
+g\bigl(\alpha_{jj}\eta_1+\alpha_j^2\eta_2\bigr)
\Bigr],
\end{aligned}
\label{eq:dxxphi}
\end{equation}
with $\alpha_{jj}=g_{jj}$.
Higher-order derivatives, including the third derivative used in KdV-type residuals, are obtained by the same chain and product rules applied to the Duhamel representation.
The transcendental quantities in the assembly are limited to $\tanh$, $\exp$, and the removable-limit functions $\eta_m$.

In multiple dimensions, the diagonal Hessian entries are obtained from \eqref{eq:dxxphi}, and the Laplacian is their trace.
Mixed derivatives, when required by a PDE residual, are obtained by applying the same chain rule to different coordinate directions.
The PoU-localized trial function $\psi_k(\bx)\phi_{i,k}(\bx,t)$ has space derivatives given by Leibniz; in particular,
\begin{equation}
\begin{aligned}
\nabla_\bx\!\bigl(\psi_k\phi_{i,k}\bigr)&=(\nabla_\bx\psi_k)\phi_{i,k}+\psi_k\nabla_\bx\phi_{i,k},\\
\Delta_\bx\!\bigl(\psi_k\phi_{i,k}\bigr)&=(\Delta_\bx\psi_k)\phi_{i,k}+2\,\nabla_\bx\psi_k\!\cdot\!\nabla_\bx\phi_{i,k}+\psi_k\Delta_\bx\phi_{i,k},
\end{aligned}
\label{eq:leibniz-pou}
\end{equation}
and the temporal derivative commutes with multiplication by $\psi_k$.
Because the bump \eqref{eq:pou-bump} has compact support inside its patch, $\psi_k\phi_{i,k}$ vanishes outside $\Omega_k$, and the assembled matrix $\bA$ in \eqref{eq:lin-system} has a natural block-sparse structure in the patch index $k$ whenever the collocation points are clustered by patch.
The assembly below is stated in dense form, while this patch-index sparsity provides an implementation-level reduction in constants.
L-RFM-Global omits the Leibniz factors involving $\psi_k$ and uses the same formulas for $\phi_i$ directly.
Thus residual assembly requires neither a numerical ODE solve nor automatic differentiation through an ODE integrator.

Complex-valued equations use the same trial functions in real-block form.
Writing $u=u_R+iu_I$,
\begin{equation}
u_R=\sum_q c_q^R\Phi_q(\bx,t),\qquad
u_I=\sum_q c_q^I\Phi_q(\bx,t),
\label{eq:complex-ansatz}
\end{equation}
and the coupled real system is assembled with the same derivative formulas; nonlinear terms are handled by the Picard linearization described below.

\subsection{Weighted residual least-squares solver and block marching}
\label{subsec:solver}

This subsection turns the fixed L-RFM trial space into an algebraic residual system.
Once the features have been sampled, the readout coefficients are the unknowns.
Linear problems require one weighted least-squares solve, while nonlinear problems and long windows use Picard linearization and block marching around the same linear algebraic core.

On a current solve window of length $T_{\rm win}$, choose interior, initial, and boundary collocation sets
\[
\mathcal{X}_{\mathrm{int}}\subset\Omega\times(0,T_{\rm win}],\qquad
\mathcal{X}_{\mathrm{IC}}\subset\overline{\Omega}\times\{0\},\qquad
\mathcal{X}_{\mathrm{BC}}\subset\partial\Omega\times(0,T_{\rm win}],
\]
with cardinalities $N_\mathrm{int}$, $N_\mathrm{IC}$, and $N_\mathrm{BC}$.
For a whole-window solve, $T_{\rm win}=T$; for a marching block, $t\in[0,T_{\rm win}]$ is the local block time and data are evaluated at the corresponding physical time.
In the experiments, the interior points are drawn by Latin hypercube sampling, while the initial and boundary sets use tensor-product grids; the construction requires collocation sets that define an overdetermined residual system.
For linear $\cN$, every row of the residual is linear in $\bc$.
At an interior point, the $q$th column entry is obtained by applying $\partial_t+\cN$ to $\Phi_q$ and evaluating at that point, while the right-hand side is the forcing value.
Initial and boundary rows are assembled by evaluating $\Phi_q$ and $\cB[\Phi_q]$, respectively.
The weighted stacked system is
\begin{equation}
\begin{aligned}
\bA\bc&\approx\bb,\\
\bA&=\bigl[\lambda_\mathrm{int}\bA^\mathrm{int};\,
\lambda_\mathrm{IC}\bA^\mathrm{IC};\,
\lambda_\mathrm{BC}\bA^\mathrm{BC}\bigr],\\
\bb&=\bigl[\lambda_\mathrm{int}\bb^\mathrm{int};\,
\lambda_\mathrm{IC}\bb^\mathrm{IC};\,
\lambda_\mathrm{BC}\bb^\mathrm{BC}\bigr],
\end{aligned}
\label{eq:lin-system}
\end{equation}
The readout vector is computed from the weighted least-squares problem
\begin{equation}
\bc=\underset{\mathbf{d}\in\R^P}{\operatorname{arg\,min}}\,
\|\bA\mathbf{d}-\bb\|_2,
\label{eq:weighted-ls}
\end{equation}
using a truncated-SVD routine.
We use the block weights
\[
\lambda_\mathrm{int}=1,\qquad
\lambda_\mathrm{IC}=\sqrt{N_\mathrm{int}/N_\mathrm{IC}},\qquad
\lambda_\mathrm{BC}=\sqrt{N_\mathrm{int}/N_\mathrm{BC}},
\]
omitting any empty block.
These weights balance the aggregate contribution of interior, initial, and boundary blocks when their row counts differ and are kept fixed across the benchmark configurations.
After row weighting, singular values of $\bA$ below $\sigma_{\max}\sigma_{\mathrm{cut}}$ are discarded, with $\sigma_{\mathrm{cut}}=10^{-12}$ unless otherwise stated.

For nonlinear problems, Picard iteration freezes the previous iterate in the coefficients or source terms.
At step $\ell$, $\cN[u]$ is replaced by a problem-specific operator $\cN_{\rm lin}[u^{(\ell)};u^{(\ell-1)}]$ that is linear in the current readout.
The residual solved at step $\ell$ is therefore
\begin{equation}
R^{(\ell)}(\bc;\bx,t)
=\partial_t u^{(\ell)}(\bx,t)
+\cN_{\rm lin}\!\bigl[u^{(\ell)};u^{(\ell-1)}\bigr](\bx,t)
-f(\bx,t),
\label{eq:picard-residual}
\end{equation}
in the same weighted block form as \eqref{eq:lin-system}.
For example, Burgers-type products freeze one factor at $u^{(\ell-1)}$ while keeping the other factor linear in $u^{(\ell)}$.
The iteration stops when the relative update $\|u^{(\ell)}-u^{(\ell-1)}\|/\|u^{(\ell)}\|$ falls below a prescribed tolerance $\delta$ or when the maximum Picard count is reached.
The same residual assembly and stopping rule are used across all nonlinear experiments.

For long time horizons, $[0,T]$ is split into $B$ blocks $[T_b,T_{b+1}]$; the case $B=1$ is the standard whole-window solve.
On block $b$, $T_{\rm win}=T_{b+1}-T_b$, a new feature family is sampled on the local time interval, and the terminal trace from the previous block is used as the next initial condition.
Block marching extends the same residual least-squares construction to long or strongly nonlinear intervals by restarting the local time coordinate and passing solution traces between consecutive blocks.

\subsection{Algorithmic summary and computational cost}
\label{subsec:algorithm}

Algorithm~\ref{alg:lrfm} summarizes the solver workflow.
When $B=1$, the marching loop reduces to a single whole-window solve.

\begin{algorithm}[!htbp]
\small
\DontPrintSemicolon
\KwIn{IBVP data; variant $\nu\in\{\mathrm{Local},\mathrm{Global}\}$; feature counts $M$ per patch for Local or $P$ global features for Global; PoU data if local; block partition; collocation sets and weights; Picard parameters $(\delta,\{L_b\}_{b=0}^{B-1})$; SVD cutoff.}
\KwOut{Approximate solution $u_{\lr}$ on $\overline{\Omega}\times[0,T]$.}
\For{$b=0,\dots,B-1$ on block $[T_b,T_{b+1}]$}{
  Set $T_{\rm win}=T_{b+1}-T_b$ and use local block time $t_{\rm loc}\in[0,T_{\rm win}]$\;
  Set the block initial trace from $u_0$ or from the previous block\;
  \If{$\nu=\mathrm{Local}$}{
    Sample $M$ primitives on each patch and set $\Phi_q=\psi_k\phi_{i,k}$, $q=(k,i)$\;
  }
  \ElseIf{$\nu=\mathrm{Global}$}{
    Sample $P$ primitives on the whole domain and set $\Phi_q=\phi_q$\;
  }
  Evaluate $\Phi_q$ and the residual derivatives; apply the PoU product rule if local\;
  Set $L_b=1$ for a linear problem and use the prescribed maximum otherwise\;
  Initialize $u^{(0)}$ from the prescribed Picard initial guess\;
  \For{$\ell=1,\dots,L_b$}{
    Assemble the weighted residual system \eqref{eq:lin-system}, using Picard linearization when needed\;
    Solve \eqref{eq:weighted-ls} by truncated SVD to obtain $\bc^{(\ell)}$\;
    Form the readout solution, using the real-block form when needed\;
    \If{the relative update is below $\delta$}{\textbf{break}\;}
  }
  Record $u_{\lr}^{(b)}\leftarrow u^{(\ell)}$ on the current block\;
}
\Return{$u_{\lr}$ obtained by concatenating all block solutions}\;
\caption{Weighted residual least-squares algorithm for the L-RFM variants.}
\label{alg:lrfm}
\end{algorithm}

Let $N=N_\mathrm{int}+N_\mathrm{IC}+N_\mathrm{BC}$ and let $P$ be the total readout dimension on one solve window.
For L-RFM-Local, $P=KM$; for L-RFM-Global, $P$ is the number of global liquid features.
Feature evaluation and matrix assembly cost $O(NP)$ work per Picard step.
For a dense overdetermined system with $N\gtrsim P$, the SVD-based LS solve costs $O(NP^2+P^3)$; if $P>N$, the corresponding dense cost is $O(N^2P+N^3)$.
The benchmark solves use overdetermined systems, so the former regime is the relevant one.
With $L=\max_b L_b$, the dense LS phase for $B$ marching blocks scales as $O(BL\,(NP^2+P^3))$ in the constant-size case.
If $N$, $P$, or the Picard count varies across blocks, the cost is the corresponding sum over blocks.
These asymptotic LS scalings match localized ST-RFM at matched $(K,M,N)$ and global methods at matched total $P$.
Liquid kernels increase the feature-evaluation constant through the evaluation of the closed-form responses and their derivatives, while the LS order is unchanged.
The stated costs use dense assembly; patch sparsity reduces implementation constants and is separate from the asymptotic comparison.
Memory is dominated by the dense residual matrix at $O(NP)$ entries.

\section{Trial-Space Expressivity and Temporal-Scale Representation}
\label{sec:theory}

For evolutionary PDEs, the expressivity question for a fixed-feature space--time solver has both an infinite- and a finite-dimensional side.
At the level of trial-space closure, the relevant issue is whether the local and global L-RFM families retain the approximation scope expected of random-feature collocation.
At the level of an assembled finite system, the relevant issue is whether sampled relaxation scales contribute distinct temporal directions rather than merely reparameterizing static activations.
This section addresses these two points.
We first prove density of the deterministic trial spaces through a separable ridge--exponential subfamily contained in the closure of the liquid feature family.
We then use a reduced temporal collocation model to show that distinct relaxation rates generate independent temporal columns at finite feature count.
Together, these results provide the theoretical basis for interpreting the finite-feature and ablation studies in Section~\ref{sec:experiments}.

\subsection{Density through a separable liquid subfamily}
\label{subsec:density-thm}

Let $Q=\overline{\Omega}\times[0,T]$ with $\Omega\subset\R^d$ bounded.
We consider the deterministic spans generated by the L-RFM features:
\[
\cA_{\mathrm{loc}}=\bigcup_{M\ge 1}\Bigl\{\sum_{k=1}^K\sum_{i=1}^M c_{i,k}\,\psi_k(\bx)\,\phi_{\theta_{i,k},k}(\bx,t):c_{i,k}\in\R,\,\theta_{i,k}\in\Theta_k\Bigr\},
\]
and
\[
\cA_{\mathrm{glob}}=\bigcup_{P\ge 1}\Bigl\{\sum_{i=1}^P c_i\,\phi_{\theta_i}(\bx,t):c_i\in\R,\,\theta_i\in\Theta_{\mathrm{glob}}\Bigr\},
\]
where $\Theta_k$ and $\Theta_{\mathrm{glob}}$ denote the admissible parameter sets for the local and global feature families.
The localized space uses the PoU factors and local affine coordinates; the global space uses the same closed-form liquid primitive with one global affine normalization and no PoU factor.

\begin{assumption}[Patch geometry and admissible parameters]
\label{assump:setting}
\hfill
\begin{enumerate}[label=(A\arabic*)]
\item\label{a-pou} $\{\psi_k\}_{k=1}^K\subset C(\overline{\Omega})$ is a continuous partition of unity, $\psi_k\ge 0$, $\sum_k\psi_k\equiv 1$, and each patch support $X_k=\mathrm{supp}\,\psi_k$ is compact.
On each patch, $\bz_k(\bx)=D_k^{-1}(\bx-\bx_k^c)$ with nonsingular diagonal $D_k$.
\item\label{a-tau} There exist $0<\tau_{\min}<\tau_{\max}<\infty$ and a fixed amplitude $A_0$ such that the closure of each $\Theta_k$ contains the slice
\begin{equation}
\begin{aligned}
(\tau,A_0,\bw,b,\bw^0,b^0)\quad\text{with}\quad
&\tau\in[\tau_{\min},\tau_{\max}],\quad \bw=\bzero,\quad b=0,\\
&(\bw^0,b^0)\in\R^d\times\R .
\end{aligned}
\label{eq:admissible-subset}
\end{equation}
The value of $A_0$ is immaterial on this slice because $g\equiv0$.
\item\label{a-global} The closure of $\Theta_{\mathrm{glob}}$ contains the same parameter slice as in \ref{a-tau}, with the global affine coordinate normalization in place of $\bz_k$.
\end{enumerate}
\end{assumption}

These assumptions are stated for the closures of admissible deterministic parameter sets.
It is sufficient for the separable slice to lie in the parameter-set closure, since the feature map is continuous in its parameters.
For random feature draws, the corresponding almost-sure statement follows when the closures of the sampling-law supports contain the same slice.
The bounded sampling ranges used in the computations are finite-dimensional realizations of this admissible family; for any fixed target and tolerance, the density argument requires only finitely many parameters from the slice, and sufficiently broad bounded ranges contain those approximants.

The key observation is that the closure of the liquid feature family contains a separable subfamily.
On the parameter slice $\bw=\bzero$ and $b=0$,
\begin{equation}
g_{i,k}\equiv0,\qquad s_{i,k}\equiv0,\qquad E_{i,k}(\bx,t)=e^{-t/\tau_{i,k}},
\label{eq:subdict-slice}
\end{equation}
and the closed-form feature reduces to
\begin{equation}
\phi_{\theta,k}(\bx,t)\Big|_{\bw=\bzero,b=0}
=\tanh\!\bigl((\bw^0)^\top\bz_k(\bx)+b^0\bigr)\,e^{-t/\tau}.
\label{eq:subdict}
\end{equation}
Thus the density argument is carried by a separable ridge--exponential component of the liquid feature family.

\begin{theorem}[Trial-space density through a separable subfamily]
\label{thm:lrfm-universal}
Under Assumption~\ref{assump:setting}, both deterministic L-RFM trial spaces $\cA_{\mathrm{glob}}$ and $\cA_{\mathrm{loc}}$ are dense in $C(Q)$ in the supremum norm, and consequently in $L^p(Q)$ for every $1\le p<\infty$.
\end{theorem}

\begin{proof}[Proof outline]
The slice \eqref{eq:subdict} gives products of spatial ridge functions and temporal exponentials.
Classical sigmoid universal approximation gives density of the ridge span on every compact patch, while a Laplace-transform argument gives density of finite exponential sums $\sum_i c_i e^{-\alpha_i t}$ on $[0,T]$ for $\alpha_i$ in any compact positive interval.
Stone--Weierstrass then gives density of the tensor-product span on each $X_k\times[0,T]$ and, for the global trial space, on $\overline{\Omega}\times[0,T]$.
For the localized trial space, approximating the target on each $X_k\times[0,T]$ to accuracy $\varepsilon$ gives a global error bounded by $\sum_k\psi_k\varepsilon=\varepsilon$ because $\psi_k\ge0$ and $\sum_k\psi_k=1$.
The full proof is given in \ref{app:density-proof}.
\end{proof}

Theorem~\ref{thm:lrfm-universal} establishes the approximation scope of the deterministic L-RFM trial spaces.
The argument uses a separable ridge--exponential subfamily, so the role of sampled relaxation scales at finite feature count is examined separately below.

\subsection{Finite-feature temporal rank calculation}
\label{subsec:finite-scale}

To isolate the effect of the $\tau$ law, we consider a reduced temporal collocation model with spatial dependence and the PDE residual operator removed.
This model captures the algebraic contribution of multiple relaxation rates before the full residual system is assembled.

\begin{proposition}[Linear independence of sampled relaxation columns]
\label{prop:scale-identifiability}
Let $0\le t_1<\cdots<t_N\le T$ and let $\lambda_j=1/\tau_j>0$ be $J\le N$ distinct relaxation rates.
Define the temporal feature matrix $V\in\R^{N\times J}$ by
\[
V_{n j}=e^{-\lambda_j t_n}.
\]
Then $V$ has full column rank.
Consequently, any grid function of the form
\[
y_n=\sum_{j=1}^J a_j e^{-\lambda_j t_n},\qquad n=1,\ldots,N,
\]
is represented exactly by $J$ exponential features with those rates.
If all features are restricted to a single rate $\lambda_0$, the best one-scale approximation error is
\[
\inf_{c\in\R}\left\|\mathbf y-c\mathbf v_{\lambda_0}\right\|_2,
\qquad
(\mathbf v_{\lambda_0})_n=e^{-\lambda_0 t_n},
\]
which is positive unless $\mathbf y$ is proportional to $\mathbf v_{\lambda_0}$ on the collocation grid.
\end{proposition}

The proof is the standard Chebyshev-system argument for exponentials and is included in \ref{app:scale-proof}.
Proposition~\ref{prop:scale-identifiability} is an algebraic rank statement for the reduced temporal collocation model.
Conditioning remains governed by the separation of the sampled rates and by the temporal collocation grid: if the rates are clustered or the grid does not resolve their separation, $V$ can be ill-conditioned.
Here ``resolvable'' means producing numerically distinguishable columns on the chosen collocation grid.
In this temporal model, a fixed-$\tau$ ablation removes the independent exponential columns supplied by multiple sampled rates, whereas sampled relaxation scales give a finite-dimensional mechanism for representing multiple temporal modes.
This proposition provides the theoretical companion to the multi-scale temporal benchmark and the $\tau$-law ablation in Section~\ref{subsec:ablations}.

\begin{remark}[Connection to the numerical study]
The two results above serve different roles in the overall argument.
The density theorem shows that the local and global deterministic trial spaces retain the approximation scope expected from random-feature collocation.
The temporal-rank calculation identifies a finite-dimensional mechanism by which sampled relaxation scales supply independent temporal directions.
In the full PDE discretization, this mechanism is coupled with spatial localization, residual weighting, SVD regularization, and Picard linearization.
The numerical study therefore tests how the temporal-scale mechanism appears after the complete residual system has been assembled.

The multi-scale temporal benchmark and $\tau$ ablation probe the relaxation-scale mechanism; the local/global comparisons separate temporal representation from spatial support; and the conditioning tables characterize the assembled least-squares systems used in the experiments.
\end{remark}

\section{Numerical Experiments}
\label{sec:experiments}

This section evaluates L-RFM as a meshless space--time least-squares collocation method for evolutionary PDEs with distinct temporal and spatial scales.
The numerical study is organized around three questions: whether the analytic feature derivatives and PoU assembly are implemented consistently, whether liquid random features improve finite-feature accuracy over static random-feature trial spaces under matched readout dimension, and what evidence isolates the role of temporal-scale representation.
We begin with known-solution verification problems, then report matched-$P$ comparisons against static localized and global baselines, and finally examine temporal-scale tests, conditioning, multidimensional assembly, long-time marching, and wall-clock cost.
The supplementary material in \ref{app:configurations} records the benchmark configurations, reference solvers, row weights, SVD thresholds, and supplementary conditioning results.

\subsection{Numerical protocol and matched-capacity design}
\label{subsec:setup}

Unless otherwise stated, experiments are run in double precision on CPU using a server equipped with an Intel(R) Xeon(R) Platinum 8476C processor.
No GPU acceleration is used.
When timing is reported, wall-clock time includes feature evaluation, matrix assembly, and the dense least-squares solve on this machine unless otherwise stated.
For each solve we record relative $L^2$ and $L^\infty$ errors on a fixed evaluation grid,
\begin{equation}
e_2=\frac{\|u_{\text{num}}-u_\star\|_{\ell^2(\text{grid})}}{\|u_\star\|_{\ell^2(\text{grid})}},\qquad
e_\infty=\frac{\|u_{\text{num}}-u_\star\|_{\ell^\infty(\text{grid})}}{\|u_\star\|_{\ell^\infty(\text{grid})}},
\label{eq:metrics}
\end{equation}
and, where relevant, the LS condition number, effective rank, residual norm, and wall-clock time.
For complex-valued or multi-component systems, norms are taken after real-block concatenation.
Known-solution problems use analytic or prescribed solutions; nonlinear problems use finer-grid Fourier pseudospectral references unless a closed form is available.
Reference resolutions and time-stepping choices are fixed for each benchmark before method comparison.
Condition numbers are singular-value ratios of the row-weighted LS matrix, and effective ranks use the solver's relative SVD threshold.
All numerical experiments are run with five independent random seeds; random features and Latin hypercube collocation points are redrawn independently, and quantitative table entries and error bands report mean $\pm$ sample standard deviation over those five runs.

The proposed solvers are L-RFM-Local and L-RFM-Global.
The static baselines are the localized ST-RFM-SoV/STC trial spaces of \cite{chen2023strfm} and the global PIELM feature family of \cite{dwivedi2020pielm}.
L-RFM-Global is the non-localized member of our frozen-feature family, not the trained hidden-state liquid solver of \cite{sun2024piln}.
All baseline comparisons use a matched-capacity protocol.
Localized methods use $P=KM$ readout coefficients and global methods use the same total number $P$ of global features.
On each benchmark row, the PDE, readout dimension, collocation sets, reference and evaluation grids, row weights, SVD threshold, Picard settings, stopping criteria, and repeat rules are fixed across methods.
Thus the comparisons isolate the change in trial functions, static versus liquid features and local versus global support, with sampling size, solver tolerance, and reference resolution fixed across methods.
Table~\ref{tab:method-hyperparameters} makes this comparison structure explicit: each row changes one trial-space component while keeping the residual assembly, solver settings, and evaluation protocol fixed.

\begin{center}
  \centering
  \captionof{table}{Matched-capacity rules for isolating trial-space effects.}
  \label{tab:method-hyperparameters}
  \TableFont
  \setlength{\tabcolsep}{3pt}
  \begin{tabular}{@{}>{\raggedright\arraybackslash}p{0.24\linewidth}>{\raggedright\arraybackslash}p{0.36\linewidth}>{\raggedright\arraybackslash}p{0.32\linewidth}@{}}
    \toprule
    Comparison & Held fixed & Tested change \\
    \midrule
    Local static vs L-RFM-Local & Same $K$, $M$, total $P=KM$, collocation, row weights, SVD, and nonlinear iteration. & Static space-time ridges are replaced by liquid temporal features at fixed local geometry. \\
    PIELM vs L-RFM-Global & Same global $P$, collocation, row weights, SVD, and evaluation grid. & A static global ridge trial space is replaced by a global liquid trial space. \\
    L-RFM-Local vs L-RFM-Global & Same total $P$ and liquid sampling law, with or without PoU localization. & Spatial localization is separated from the shared liquid temporal response. \\
    Ablation variants & Same PoU, $P$, collocation, row weights, Picard settings, and SVD threshold. & The liquid response is removed or the $\tau$ distribution is collapsed. \\
    \bottomrule
  \end{tabular}
\end{center}

\subsection{Assembly verification and collocation sensitivity}
\label{subsec:verification-convergence}

This subsection verifies implementation and resolution effects before the primary method-comparison experiments.
On exact or prescribed-solution heat problems, we verify that the analytic feature derivatives, row-weighted LS assembly, PoU support, collocation row counts, and tensor-product Laplacian are assembled consistently before nonlinear benchmarks are used to compare methods.
The rows in this subsection check the discretized machinery before liquid and static trial spaces are compared.

\subsubsection{Heat equation in one dimension and collocation sensitivity}

The first verification problem is the one-dimensional heat equation with the exact decay
$u_\star(x,t)=e^{-\pi^2t}\sin(\pi x)$ on $[-1,1]\times[0,0.25]$.
It solves $\partial_t u=\partial_{xx}u$ with homogeneous Dirichlet boundary conditions and initial data $u(x,0)=\sin(\pi x)$.
For this localized sweep, $K$ is the number of PoU patches, $M$ is the number of features per patch, and $P=KM$ is the readout dimension.
Sweeping the PoU count and the per-patch feature count checks the closed-form derivative implementation and LS residual assembly before nonlinear or dispersive effects are introduced.
Across five random seeds, the relative $L^2$ error falls from $(4.43\pm1.44){\times}10^{-1}$ at $P=25$ to $(5.99\pm2.84){\times}10^{-7}$ at $K=4$, $M=100$.
Increasing to $K=8$, $M=100$ gives $(2.01\pm1.75){\times}10^{-5}$ under the same solver protocol, consistent with rank and conditioning becoming limiting factors in the overparameterized LS system.
Figure~\ref{fig:heat-convergence} shows the corresponding PoU and feature-count sweeps.
A separate collocation sensitivity check then fixes $K=4$, $M=100$ and varies the total LS row count from $2P$ to $6P$, thereby testing overdetermination at a fixed trial space.
Table~\ref{tab:collocation-sensitivity} shows that the effective rank remains below $P$ under the common SVD truncation rule, but the $2P$ row setting already gives five-seed mean error $(1.60\pm1.56){\times}10^{-5}$; increasing to $3P$--$6P$ moves the mean error into the range $1.53{\times}10^{-7}$--$5.39{\times}10^{-6}$.
The large raw condition numbers also motivate the SVD-truncated least-squares solve used throughout the paper; the reported ranks are retained numerical ranks under the common truncation threshold.
These results support using overdetermined collocation systems with several rows per readout coefficient in the main experiments.

\begin{center}
\centering
\captionof{table}{Effect of collocation overdetermination on the 1D heat-equation solve. Rank denotes the SVD-retained numerical rank under the common truncation rule.}
\label{tab:collocation-sensitivity}
\TableFont
\setlength{\tabcolsep}{4pt}
\begin{tabular}{@{}ccccccc@{}}
\toprule
\multicolumn{3}{c}{Collocation design} & \multicolumn{2}{c}{Error} & \multicolumn{2}{c}{Linear algebra} \\
\cmidrule(lr){1-3}\cmidrule(lr){4-5}\cmidrule(l){6-7}
$N/P$ & $N$ & $(N_{\mathrm{int}},N_{\mathrm{IC}},N_{\mathrm{BC}})$ & Rel.\ $L^2$ & Rel.\ $L^\infty$ & $\kappa(\mathbf A)$ & Rank \\
\midrule
    $2$ & 800 & $(560,120,60)$ & $(1.60\pm1.56){\times}10^{-5}$ & $(1.80\pm1.21){\times}10^{-5}$ & $(3.06\pm3.07){\times}10^{14}$ & 361--378 \\
    $3$ & 1200 & $(840,180,90)$ & $(5.39\pm10.46){\times}10^{-6}$ & $(7.11\pm13.15){\times}10^{-6}$ & $(1.24\pm1.19){\times}10^{14}$ & 369--385 \\
    $4$ & 1600 & $(1120,240,120)$ & $(2.90\pm2.03){\times}10^{-7}$ & $(4.56\pm2.11){\times}10^{-7}$ & $(8.90\pm8.83){\times}10^{13}$ & 370--387 \\
    $6$ & 2400 & $(1680,360,180)$ & $(1.53\pm1.18){\times}10^{-7}$ & $(2.01\pm1.57){\times}10^{-7}$ & $(6.76\pm6.26){\times}10^{13}$ & 374--391 \\
    \bottomrule
\end{tabular}
\end{center}

\begin{center}
  \centering
  \maybeincludegraphics[width=0.58\linewidth]{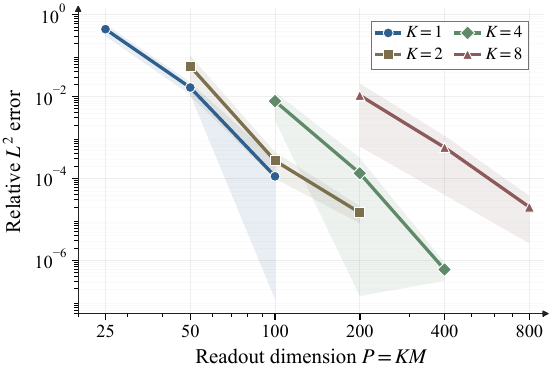}
  \captionof{figure}{Convergence of the 1D heat-equation solve under PoU and feature-count refinement. Curves and bands show mean $\pm$ sample standard deviation over five random seeds.}
  \label{fig:heat-convergence}
\end{center}

\subsubsection{Two-dimensional heat-equation verification}

This test verifies the tensor-product PoU and Laplacian assembly in two spatial dimensions using the prescribed heat solution
\[
u_\star(x_1,x_2,t)=e^{-2\pi^2t}\sin(\pi x_1)\sin(\pi x_2)
\]
on $[0,1]^2\times[0,0.12]$.
It solves $\partial_t u=\Delta u$ with homogeneous Dirichlet boundary conditions and initial data
$u(x_1,x_2,0)=\sin(\pi x_1)\sin(\pi x_2)$.
Here the PoU layout is $2\times2$, so $K=4$ and $P=4M$.
Table~\ref{tab:heat2d-check} shows that increasing the per-patch feature count monotonically reduces both relative $L^2$ and relative $L^\infty$ errors, reaching $(2.57\pm0.20){\times}10^{-3}$ and $(8.20\pm2.13){\times}10^{-3}$, respectively, at $P=240$ over five random seeds.
The monotone decrease confirms that the local-coordinate derivatives, PoU products, and Laplacian terms remain consistent after moving from one to two spatial dimensions.

\begin{center}
\centering
\captionof{table}{Two-dimensional heat-equation verification with tensor-product PoU assembly. Values are reported over five random seeds; the PoU layout is $2\times2$, so $P=4M$.}
\label{tab:heat2d-check}
\small
\setlength{\tabcolsep}{6pt}
\begin{tabular}{@{}cccc@{}}
\toprule
\multicolumn{2}{c}{Capacity} & \multicolumn{2}{c}{Error} \\
\cmidrule(lr){1-2}\cmidrule(l){3-4}
$M$ per patch & $P$ & Rel.\ $L^2$ & Rel.\ $L^\infty$ \\
\midrule
    20 & 80 & $(1.00\pm0.24){\times}10^{-1}$ & $(1.31\pm0.40){\times}10^{-1}$ \\
    30 & 120 & $(4.18\pm0.84){\times}10^{-2}$ & $(8.04\pm1.68){\times}10^{-2}$ \\
    40 & 160 & $(1.38\pm0.45){\times}10^{-2}$ & $(2.61\pm0.78){\times}10^{-2}$ \\
    60 & 240 & $(2.57\pm0.20){\times}10^{-3}$ & $(8.20\pm2.13){\times}10^{-3}$ \\
    \bottomrule
\end{tabular}
\end{center}

\FloatBarrier

\subsection{Representative one-dimensional matched-capacity comparisons}
\label{subsec:local-global-1d}

We next report a representative one-dimensional matched-capacity comparison between the proposed liquid trial spaces and static random-feature baselines.
All rows use the same collocation sets, row weights, SVD truncation threshold, nonlinear iteration settings, and evaluation grid within each problem.
Table~\ref{tab:main-baselines} compares two design choices under the same readout dimension.
The static-to-liquid replacement is assessed by comparing PIELM with L-RFM-Global and ST-RFM-SoV/STC with L-RFM-Local, while the localization choice is assessed by comparing L-RFM-Local and L-RFM-Global at matched $P$.
The four rows probe complementary mechanisms: a sharp reaction--diffusion interface (Allen--Cahn), smooth nonlinear transport (Burgers), a third-derivative dispersive residual (KdV), and complex-valued real-block assembly (NLS).
The benchmark definitions are kept explicit for reproducibility.
Allen--Cahn solves $u_t=\epsilon u_{xx}-u^3+u$ \cite{allen1979microscopic} on $[-1,1]\times[0,0.1]$ with $\epsilon=10^{-4}$, periodic boundary conditions, and $u(x,0)=0.55\sin(\pi x)+0.25\sin(2\pi x)$.
Burgers solves $u_t+uu_x=\nu u_{xx}$ on $[-1,1]\times[0,0.10]$ with $\nu=0.005/\pi$, periodic boundary conditions, and $u(x,0)=-\sin(\pi x)$.
KdV solves $u_t+6uu_x+u_{xxx}=0$ on $[-8,8]\times[0,0.5]$ with periodic boundary conditions and $u(x,0)=\frac{c}{2}\operatorname{sech}^2(\frac{\sqrt{c}}{2}(x-x_0))$, where $c=2$ and $x_0=-3$.
NLS solves $\mathrm{i}\psi_t+\frac{1}{2}\psi_{xx}+|\psi|^2\psi=0$ on $[-8,8]\times[0,0.4]$ with Dirichlet boundary conditions from the bright soliton and $\psi(x,0)=\eta\operatorname{sech}(\eta(x-x_0))\exp(\mathrm{i}v(x-x_0)+\mathrm{i}\phi_0)$, where $\eta=1$, $v=0.5$, $x_0=-2$, and $\phi_0=0$.

\begin{center}
  \centering
  \captionof{table}{Matched-readout comparison on representative one-dimensional PDE benchmarks. Entries are relative $L^2$ errors, mean $\pm$ sample standard deviation over five random seeds. Bold marks the lowest mean.}
  \label{tab:main-baselines}
  \TableFont
  \setlength{\tabcolsep}{2pt}
  \begin{tabular}{@{}lcc@{\hspace{8pt}}ccc@{\hspace{8pt}}cc@{}}
    \toprule
    & & \multicolumn{3}{c}{Static baselines} & \multicolumn{2}{c}{L-RFM variants} \\
    \cmidrule(lr){3-5}\cmidrule(l){6-7}
    Problem & $P$ & ST-RFM-SoV & ST-RFM-STC & PIELM & L-RFM-Local & L-RFM-Global \\
    \midrule
    Allen--Cahn (interface) & 400 & $(4.09\pm1.00){\times}10^{-6}$ & $(1.04\pm0.33){\times}10^{-5}$ & $(3.22\pm1.48){\times}10^{-6}$ & $\boldsymbol{(5.98\pm0.83){\times}10^{-8}}$ & $(3.78\pm1.31){\times}10^{-7}$ \\
    Burgers (smooth global) & 480 & $(7.39\pm2.94){\times}10^{-6}$ & $(4.56\pm0.77){\times}10^{-6}$ & $(1.09\pm0.72){\times}10^{-5}$ & $(3.30\pm1.68){\times}10^{-5}$ & $\boldsymbol{(1.81\pm0.80){\times}10^{-6}}$ \\
    KdV (third derivative) & 480 & $(4.63\pm1.89){\times}10^{-3}$ & $(6.16\pm1.32){\times}10^{-3}$ & $(1.60\pm1.53){\times}10^{-2}$ & $\boldsymbol{(1.61\pm0.61){\times}10^{-3}}$ & $(3.57\pm5.21){\times}10^{-3}$ \\
    NLS (complex) & 960 & $(8.93\pm2.37){\times}10^{-4}$ & $(2.30\pm1.09){\times}10^{-3}$ & $(1.82\pm1.94){\times}10^{-2}$ & $\boldsymbol{(8.64\pm4.51){\times}10^{-5}}$ & $(2.25\pm1.40){\times}10^{-4}$ \\
    \bottomrule
  \end{tabular}
\end{center}

Overall, a liquid trial space gives the lowest mean error in all four rows: L-RFM-Local for Allen--Cahn, KdV, and NLS, and L-RFM-Global for Burgers.
The pattern indicates that the proposed temporal liquid response is useful in both localized and global forms.
For the smooth periodic Burgers case, one global liquid trial space is more efficient than splitting the same readout dimension across patches; for Allen--Cahn, KdV, and NLS, the temporal response and PoU localization act together.
Figure~\ref{fig:main-1d-solutions} visualizes representative reference fields, L-RFM-Local fields, and absolute errors for the four rows in Table~\ref{tab:main-baselines}; the quantitative comparison, including the global Burgers winner, is reported in the table.

\begin{center}
  \centering
  \maybeincludegraphics[width=\FigFieldWidth]{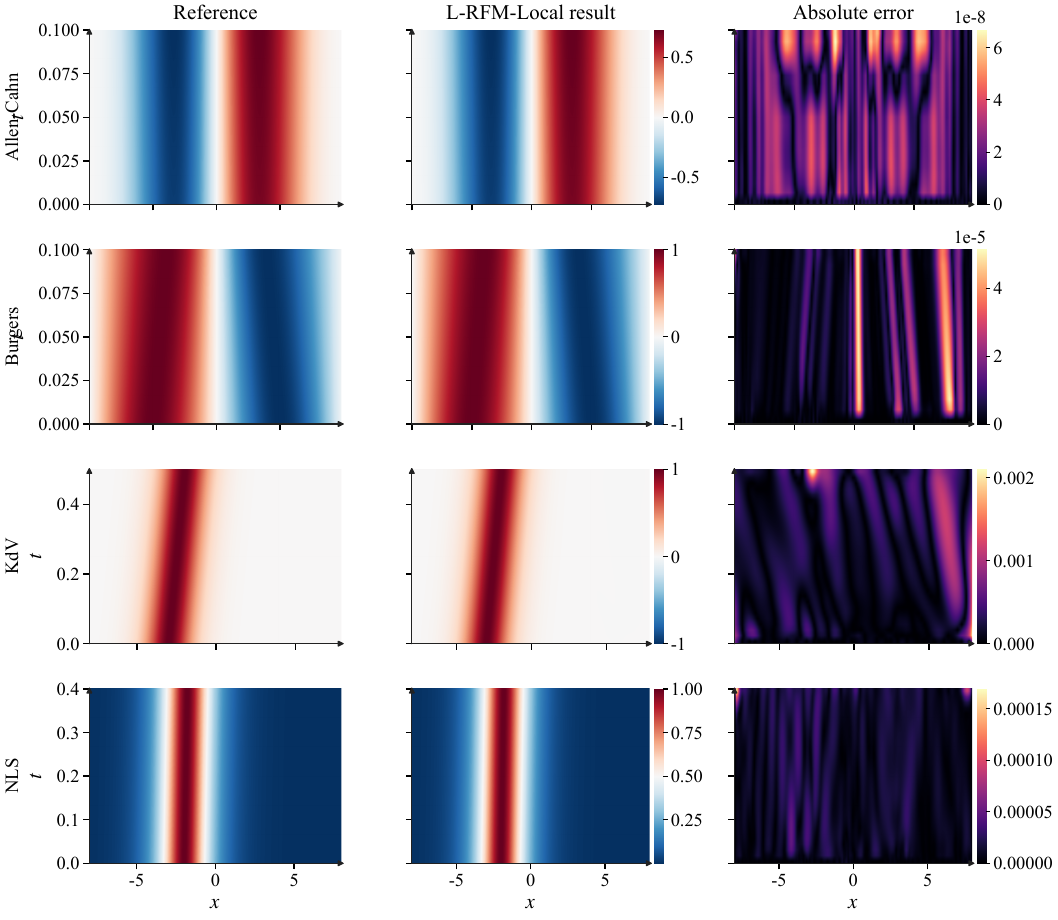}
\captionof{figure}{Representative solution fields for the one-dimensional benchmark suite. Rows show Allen--Cahn, Burgers, KdV, and NLS; columns show reference, L-RFM-Local result, and absolute error. For NLS, the plotted scalar is $|\psi|$. The figure fixes the displayed solver to L-RFM-Local to show solution scales and localized-field quality; quantitative method comparisons, including the L-RFM-Global Burgers row, are reported in Table~\ref{tab:main-baselines}.}
  \label{fig:main-1d-solutions}
\end{center}

\FloatBarrier

\subsection{Mechanism tests: temporal scales, diffusion scale, and conditioning}
\label{subsec:mechanism-evidence}

The matched-capacity comparisons motivate a closer look at the mechanisms associated with the liquid-feature gains.
This subsection isolates four questions: whether the liquid trial space represents separated temporal scales, whether the same response remains accurate as the Allen--Cahn diffusion parameter is varied, whether the gain persists under component ablations, and how the resulting row-weighted LS systems are conditioned.

\subsubsection{Prescribed multi-scale temporal test}
\label{subsec:core-temporal}

We use a prescribed multi-scale temporal solution on $[-1,1]\times[0,0.1]$ to separate temporal representation from nonlinear interface effects:
\begin{equation}
u_\star(x,t)=\sin(\pi x)\sum_{j=1}^3\alpha_j e^{-t/\tau_j},
\qquad
(\tau_1,\tau_2,\tau_3)=(1,10^{-2},10^{-4}),\quad \alpha_j=\tfrac13 .
\label{eq:mms-multiscale}
\end{equation}
The problem solves $\partial_tu-\partial_{xx}u=f(x,t)$ with periodic boundary conditions, initial data $u(x,0)=u_\star(x,0)$, and source $f=u_{\star,t}-u_{\star,xx}$.
Unlike Allen--Cahn, this problem has no nonlinear interface or spatial front: it directly probes whether the sampled trial space can represent several separated relaxation scales.
At the largest tested readout dimension ($P=640$), L-RFM-Global attains a five-seed mean relative $L^2$ error of $(4.22\pm2.43){\times}10^{-4}$, while the strongest static trial space (ST-RFM-SoV) attains $(6.33\pm3.37){\times}10^{0}$.
Across the sweep in Figure~\ref{fig:mms-convergence}, the liquid rows decrease rapidly with $P$, whereas the static rows remain at order-one or larger errors at the largest tested dimension.
Because the spatial factor is a single smooth sinusoidal mode, the L-RFM-Global row gives a direct test of the temporal response without PoU effects.
This experiment separates temporal-scale representation from the nonlinear-interface benchmarks.
The full readout-dimension sweep is shown in Figure~\ref{fig:mms-convergence}.

\begin{center}
  \centering
  \maybeincludegraphics[width=0.75\linewidth]{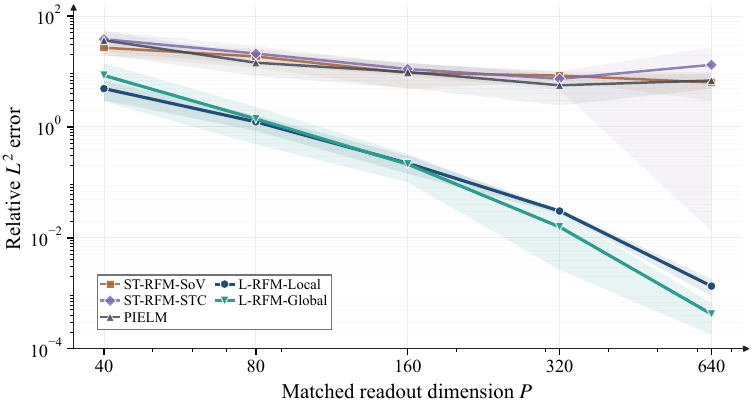}
  \captionof{figure}{Temporal-scale resolution on a prescribed multi-scale temporal problem. Curves show five-seed mean relative $L^2$ error versus $P$; shaded bands show sample standard deviation.}
  \label{fig:mms-convergence}
\end{center}

\subsubsection{Allen--Cahn diffusion-scale sweep}

We next revisit the Allen--Cahn equation and periodic initial data introduced in Section~\ref{subsec:local-global-1d}, now on $[-1,1]\times(0,0.15]$ and with the diffusion parameter varied.
This test varies the diffusion scale in the implemented Allen--Cahn model and progressively sharpens the transition region.
At fixed $P=400$, sweeping $\epsilon$ from $10^{-1}$ to $10^{-4}$ changes the diffusion scale by four decades and makes the transition region increasingly sharp; each reported row uses five random seeds.
Across five-seed means, the best liquid row spans $1.27{\times}10^{-7}$--$3.13{\times}10^{-7}$, the best static localized row spans $8.62{\times}10^{-6}$--$1.28{\times}10^{-5}$, and PIELM spans $1.57{\times}10^{-6}$--$2.22{\times}10^{-5}$.
At $\epsilon=10^{-4}$, the raw row-weighted LS condition numbers are $4.70{\times}10^{13}$ for L-RFM-Local, $2.66{\times}10^{16}$--$5.35{\times}10^{16}$ for the two static localized trial spaces, $2.59{\times}10^{15}$ for PIELM, and $1.69{\times}10^{16}$ for L-RFM-Global.
All solves use the common SVD truncation threshold, so these condition numbers measure the assembled row-weighted matrices before regularized least-squares solution.
Table~\ref{tab:stiff} reports the five-seed error sweep, and Figure~\ref{fig:stiff} plots the error and raw conditioning trends.
At $\epsilon=10^{-1}$, L-RFM-Global and L-RFM-Local are statistically close: the global row has a slightly lower mean, but the two sample-standard-deviation bands overlap.
This behavior is consistent with a broad, spatially smooth response for which PoU localization contributes little, while splitting a fixed $P$ across patches reduces the number of features available to each local fit.
As $\epsilon$ decreases and the transition region sharpens, the localized liquid row becomes the stronger matched-$P$ choice.
Figure~\ref{fig:stiff-solution} gives a representative Allen--Cahn solution field at $\epsilon=10^{-4}$.

\begin{center}
  \centering
  \captionof{table}{Diffusion-scale robustness on Allen--Cahn at matched readout dimension. Entries are relative $L^2$ errors at $P=400$, reported as mean $\pm$ sample standard deviation over five random seeds. Bold marks the lowest mean.}
  \label{tab:stiff}
  \TableFont
  \begin{tabular}{@{}lccc@{\hspace{8pt}}cc@{}}
    \toprule
    & \multicolumn{3}{c}{Static baselines} & \multicolumn{2}{c}{L-RFM variants} \\
    \cmidrule(lr){2-4}\cmidrule(l){5-6}
    $\epsilon$ & ST-RFM-SoV & ST-RFM-STC & PIELM & L-RFM-Local & L-RFM-Global \\
    \midrule
    $10^{-1}$ & $(1.28\pm0.24){\times}10^{-5}$ & $(1.95\pm0.85){\times}10^{-5}$ & $(1.57\pm1.30){\times}10^{-6}$ & $(2.92\pm1.77){\times}10^{-7}$ & $\boldsymbol{(2.89\pm1.10){\times}10^{-7}}$ \\
    $10^{-2}$ & $(1.00\pm0.11){\times}10^{-5}$ & $(1.39\pm0.41){\times}10^{-5}$ & $(9.83\pm6.85){\times}10^{-6}$ & $\boldsymbol{(1.27\pm0.36){\times}10^{-7}}$ & $(3.58\pm1.38){\times}10^{-7}$ \\
    $10^{-3}$ & $(8.62\pm0.62){\times}10^{-6}$ & $(1.16\pm0.28){\times}10^{-5}$ & $(1.37\pm0.29){\times}10^{-5}$ & $\boldsymbol{(3.13\pm1.83){\times}10^{-7}}$ & $(1.06\pm0.42){\times}10^{-6}$ \\
    $10^{-4}$ & $(1.11\pm0.24){\times}10^{-5}$ & $(1.14\pm0.51){\times}10^{-5}$ & $(2.22\pm2.68){\times}10^{-5}$ & $\boldsymbol{(3.09\pm1.92){\times}10^{-7}}$ & $(8.67\pm6.34){\times}10^{-7}$ \\
    \bottomrule
  \end{tabular}
\end{center}

\begin{center}
  \centering
  \maybeincludegraphics[width=\FigDiagnosticWidth]{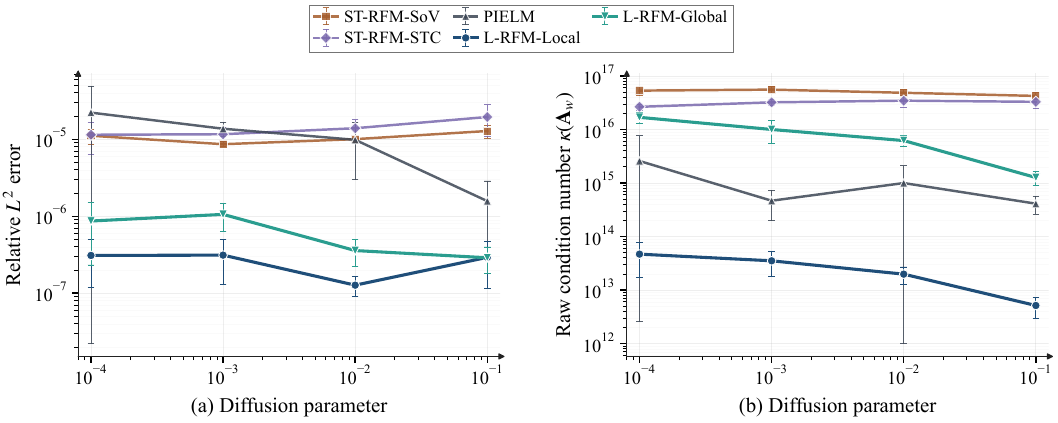}
  \captionof{figure}{Accuracy and conditioning trends across the Allen--Cahn diffusion-scale sweep. Left: mean relative $L^2$ error. Right: raw row-weighted LS condition number. Error bars show sample standard deviation over five random seeds.}
  \label{fig:stiff}
\end{center}

\begin{center}
  \centering
  \maybeincludegraphics[width=\FigFieldWidth]{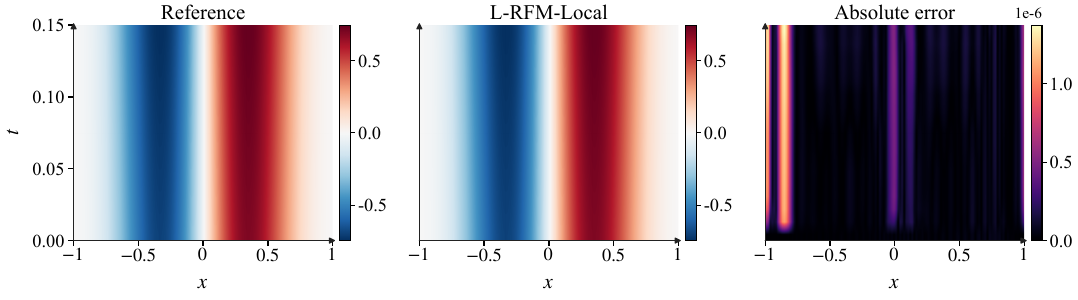}
  \captionof{figure}{Representative stiff Allen--Cahn solution field at $\epsilon=10^{-4}$.}
  \label{fig:stiff-solution}
\end{center}

\subsubsection{Ablation of temporal response and time-scale sampling}
\label{subsec:ablations}

This subsection isolates two components of the liquid trial space: the ODE-driven temporal response and the sampling law for the time scale $\tau$.
The test separates the temporal mechanism from the choice of comparison method.
In the static-feature ablation, we replace each liquid feature by its $t=0$ response and set its time derivative to zero, while keeping the PoU, random-sampling protocol, collocation sizes, row weights, Picard iteration, and SVD solver fixed.
We now perform this ablation on Heat, Allen--Cahn, Burgers, KdV, and NLS so that the test covers the same one-dimensional mechanisms used in the matched-capacity comparison: linear diffusion, reaction--diffusion, nonlinear transport, third-order dispersion, and complex-valued real-block assembly.
Removing the temporal response consistently degrades accuracy: Heat 1D changes from $(3.28\pm2.07){\times}10^{-6}$ to $(9.13\pm0.03){\times}10^{-1}$, Allen--Cahn 1D from $(9.66\pm5.41){\times}10^{-8}$ to $(1.84\pm0.05){\times}10^{-1}$, Burgers 1D from $(1.32\pm0.71){\times}10^{-2}$ to $(3.72\pm0.03){\times}10^{-1}$, KdV 1D from $(1.63\pm0.42){\times}10^{-3}$ to $(3.66\pm0.07){\times}10^{-1}$, and NLS 1D from $(7.91\pm3.47){\times}10^{-5}$ to $(1.17\pm0.03){\times}10^{-1}$.
The common pattern is that a spatially localized but temporally static trial space cannot reproduce the time-dependent residuals at comparable accuracy, even when all sampling, PoU, and solver settings are unchanged.

The $\tau$ ablation is run with five random seeds on Allen--Cahn.
The default row samples $\tau/T$ log-uniformly on $[0.05,20]$, while the fixed-$\tau$ rows use $\tau/T=0.05$, $1$, and $20$.
The log-uniform row has error $(7.55\pm2.94){\times}10^{-8}$; the small fixed-$\tau$ row gives $(2.49\pm0.93){\times}10^{-1}$, the geometric-mean fixed-$\tau$ row gives $(3.08\pm0.56){\times}10^{-7}$, and the large fixed-$\tau$ row gives $(5.20\pm0.14){\times}10^{-8}$.
These results support the interpretation that the ODE-driven temporal response is the dominant component in this test, while the $\tau$-sampling law mainly improves robustness to time-scale misspecification.
The tuned large fixed-$\tau$ row is competitive on this Allen--Cahn setting, while the log-uniform row supplies a more robust default when the relevant relaxation scale is not known a priori.
Figure~\ref{fig:ablations} summarizes these two tests.

\begin{figure}[t]
  \centering
  \begin{subfigure}{0.48\linewidth}
    \centering
    \maybeincludegraphics[width=\linewidth]{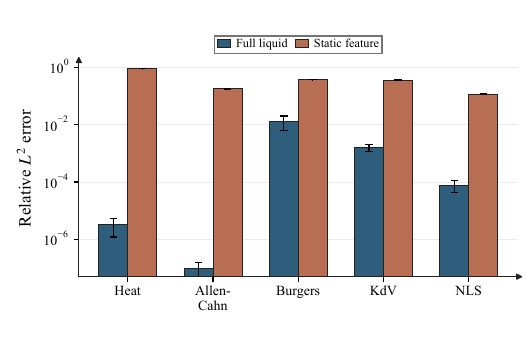}
    \caption{Static-feature ablation without the ODE-driven temporal response.}
    \label{fig:static-ablation}
  \end{subfigure}\hfill
  \begin{subfigure}{0.48\linewidth}
    \centering
    \maybeincludegraphics[width=\linewidth]{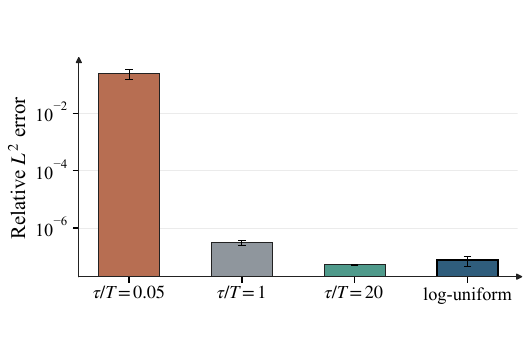}
    \caption{Time-scale sampling ablation with fixed $\tau/T$.}
    \label{fig:tau-ablation}
  \end{subfigure}
  \caption{Ablations of the liquid temporal response and time-scale sampling law. Left: removing the ODE-driven temporal response on the one-dimensional benchmark set. Right: replacing log-uniform $\tau/T\in[0.05,20]$ by fixed $\tau/T\in\{0.05,1,20\}$ on Allen--Cahn. Bars show five-seed mean relative $L^2$ error; error bars show sample standard deviation.}
  \label{fig:ablations}
\end{figure}

\subsubsection{Conditioning of matched-capacity systems}
\label{subsec:convergence-conditioning}

For the localized liquid trial space, the accuracy gains are often accompanied by smaller raw row-weighted LS condition numbers on the tested difficult rows.
These empirical measurements complement the density and temporal-rank results in Section~\ref{sec:theory} by characterizing the assembled finite-feature systems.
We report $\kappa(\mathbf A_w)$, where $\mathbf A_w$ is the row-weighted LS matrix assembled before SVD truncation.
The values are raw measurements for the untruncated weighted matrix; all solves nevertheless use the common SVD truncation threshold.
Table~\ref{tab:condition-panorama} includes the Heat and Burgers scalar PDEs used in the left panel of Figure~\ref{fig:ablations}, together with three rows that stress the least-squares geometry: small-$\epsilon$ reaction--diffusion, complex real-block assembly, and a third-derivative dispersive residual.
Across these Heat, Allen--Cahn, Burgers, NLS, and KdV rows, the L-RFM-Local matrix has a condition number two to five orders of magnitude below the static random-feature matrices at matched readout dimension.
The large L-RFM-Global condition numbers in the same rows show that conditioning and temporal expressivity are distinct finite-feature effects.
Together, the results support the view that PoU localization combined with liquid temporal features improves the finite-feature LS geometry on the tested difficult rows.

\begin{center}
  \centering
  \captionof{table}{Conditioning of selected matched-capacity least-squares systems. Entries report $\kappa(\mathbf A_w)$ for the row-weighted matrix before SVD truncation, as mean $\pm$ sample standard deviation over five random seeds; all reported solves use the common SVD truncation threshold.}
  \label{tab:condition-panorama}
  \TableFont
  \setlength{\tabcolsep}{4pt}
  \begin{tabular}{@{}lccc@{\hspace{8pt}}cc@{}}
    \toprule
    & \multicolumn{3}{c}{Static baselines} & \multicolumn{2}{c}{L-RFM variants} \\
    \cmidrule(lr){2-4}\cmidrule(l){5-6}
    Problem & ST-RFM-SoV & ST-RFM-STC & PIELM & L-RFM-Local & L-RFM-Global \\
    \midrule
    Heat 1D & $(4.41\pm1.70){\times}10^{16}$ & $(1.31\pm0.59){\times}10^{16}$ & $(6.49\pm2.01){\times}10^{16}$ & $\boldsymbol{(5.55\pm3.41){\times}10^{13}}$ & $(5.60\pm1.46){\times}10^{16}$ \\
    Allen--Cahn, $\epsilon=10^{-4}$ & $(5.35\pm1.03){\times}10^{16}$ & $(2.66\pm0.35){\times}10^{16}$ & $(2.59\pm5.11){\times}10^{15}$ & $\boldsymbol{(4.70\pm3.01){\times}10^{13}}$ & $(1.69\pm0.38){\times}10^{16}$ \\
    Burgers 1D & $(2.74\pm0.96){\times}10^{16}$ & $(5.42\pm1.80){\times}10^{15}$ & $(1.77\pm0.96){\times}10^{16}$ & $\boldsymbol{(3.69\pm4.15){\times}10^{12}}$ & $(9.15\pm2.15){\times}10^{16}$ \\
    NLS & $(2.53\pm1.47){\times}10^{16}$ & $(2.65\pm1.09){\times}10^{15}$ & $(1.33\pm0.90){\times}10^{16}$ & $\boldsymbol{(1.02\pm0.38){\times}10^{11}}$ & $(1.77\pm0.62){\times}10^{16}$ \\
    KdV & $(9.55\pm7.67){\times}10^{15}$ & $(6.23\pm3.16){\times}10^{14}$ & $(6.35\pm4.05){\times}10^{15}$ & $\boldsymbol{(2.41\pm5.38){\times}10^{12}}$ & $(2.40\pm5.28){\times}10^{16}$ \\
    \bottomrule
  \end{tabular}
\end{center}

\FloatBarrier

\subsection{Multidimensional and long-time robustness checks}
\label{subsec:robustness-limits}

After the one-dimensional matched-capacity comparisons and mechanism tests, we include two additional checks beyond the main benchmark setting.
The first examines whether the same multiscale temporal construction remains usable in a three-dimensional tensor-product assembly.
The second tests long-time block marching and records the current implementation timing behavior.
These results test robustness and scope while preserving the main comparison protocol.

\subsubsection{Three-dimensional prescribed temporal-scale check}
\label{subsec:scope-checks}

This subsection extends the prescribed multi-scale temporal problem in Section~\ref{subsec:core-temporal} to three spatial dimensions.
The test complements the one-dimensional sweep by combining the same relaxation-scale mechanism with a tensor-product PoU, a three-dimensional Laplacian, and periodic boundary matching.
We solve the forced heat equation
\begin{equation}
  u_t-\Delta u=f(\bx,t),\qquad \bx\in[-1,1]^3,\quad t\in(0,0.1],
  \label{eq:mms3d-pde}
\end{equation}
with periodic value and normal-derivative matching on opposite faces.
The exact solution is
\begin{equation}
  \begin{aligned}
  u_\star(\bx,t)&=a(\bx)q(t),\\
  a(\bx)&=1+\frac{\sin(\pi x)+\sin(\pi y)+\sin(\pi z)}{3},\\
  q(t)&=\frac{1}{3}\left(e^{-t}+e^{-10^2t}+e^{-10^4t}\right),
  \end{aligned}
  \label{eq:mms3d-exact}
\end{equation}
and $f$ is obtained by substituting \eqref{eq:mms3d-exact} into \eqref{eq:mms3d-pde}.
The nonzero mean in $a(\bx)$ prevents the test from being dominated by near-zero spatial nodes, while the three sine terms still exercise all components of the 3D Laplacian.
All rows use the same collocation set with 5200 interior points, 1000 initial points, and 160 paired samples per coordinate direction for periodic matching.
We use one matched readout dimension, $P=3840$; localized methods use a $2\times2\times2$ PoU with $M=480$ features per patch, and global methods use 3840 scalar features.
The tuning changes only generic numerical parameters, namely the random-feature capacity, random scale range, and collocation row count.

With these fixed experimental parameters, the five-seed time-space relative $L^2$ errors are $(7.38\pm3.05){\times}10^{1}$ for ST-RFM-SoV, $(1.76\pm0.78){\times}10^{2}$ for ST-RFM-STC, $(4.16\pm0.30){\times}10^{1}$ for PIELM, $(5.99\pm3.86){\times}10^{-3}$ for L-RFM-Local, and $(4.14\pm0.63){\times}10^{-2}$ for L-RFM-Global.
L-RFM-Local gives the lowest mean relative $L^2$ error, while the global liquid row is the next closest and the static rows remain several orders of magnitude larger on this test.
The result is consistent with the temporal-rank calculation in Section~\ref{subsec:finite-scale}: the prescribed response contains separated relaxation scales, and the liquid trial spaces supply corresponding temporal directions before the 3D residual is assembled.
Thus this experiment is a controlled 3D extension of the temporal-scale test, showing that the localized liquid trial space remains the most effective fixed-$P$ choice among the tested trial spaces when the same three relaxation scales are embedded in a 3D heat residual.
Figure~\ref{fig:mms3d-solution} shows representative mid-plane slices for the L-RFM-Local solution.

\begin{center}
  \maybeincludegraphics[width=0.85\linewidth]{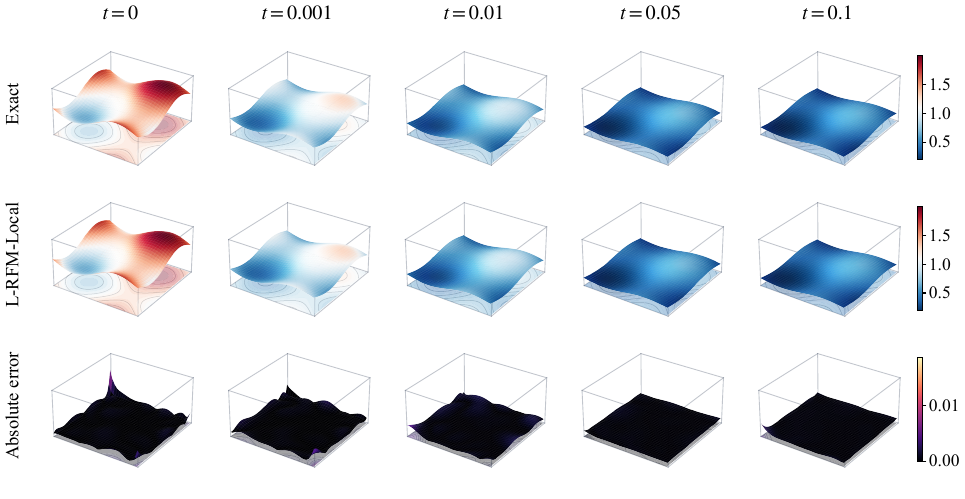}
  \captionof{figure}{Three-dimensional prescribed temporal-scale problem solved by L-RFM-Local. Mid-plane slices at $z=0.5$ are shown for one representative seed; columns show five time levels, and rows show the reference solution, L-RFM-Local solution, and absolute error.}
  \label{fig:mms3d-solution}
\end{center}

\subsubsection{Block marching on long-time Allen--Cahn}

This experiment tests whether a long time interval can be decomposed into shorter space--time LS windows without losing stability, while keeping the per-block feature count fixed.
Long time windows are therefore treated with block marching, with each block carrying its own Picard iterations and LS factorizations.
This row tests long-time stability and error propagation under a fixed per-block feature count.

Long-time Allen--Cahn on $T=2.0$ tests this operating boundary by keeping the per-block feature count fixed at $P=320$ and varying the number of blocks $B=1,4,8$.
Over five random seeds, the final-time relative $L^2$ errors are $(1.88\pm1.10){\times}10^{-1}$ for one block, $(1.23\pm0.63){\times}10^{-2}$ for four blocks, and $(8.26\pm3.03){\times}10^{-3}$ for eight blocks.
The corresponding solve times are $(1.45\pm0.01){\times}10^{1}$ s, $(4.58\pm0.35){\times}10^{1}$ s, and $(6.32\pm0.14){\times}10^{1}$ s, respectively.
Thus the main accuracy gain occurs when the long interval is first split into shorter windows; increasing from four to eight blocks gives a smaller additional gain, consistent with diminishing returns once the per-block feature count and repeated LS cost become active limitations.
Figure~\ref{fig:block-marching} plots the time-history relative $L^2$ errors for this block-count sweep, and Figure~\ref{fig:block-marching-solution} confirms that the block-marched solution preserves the interface evolution over the full time interval.

\begin{center}
  \begin{minipage}{0.72\linewidth}
  \centering
  \maybeincludegraphics[width=\linewidth]{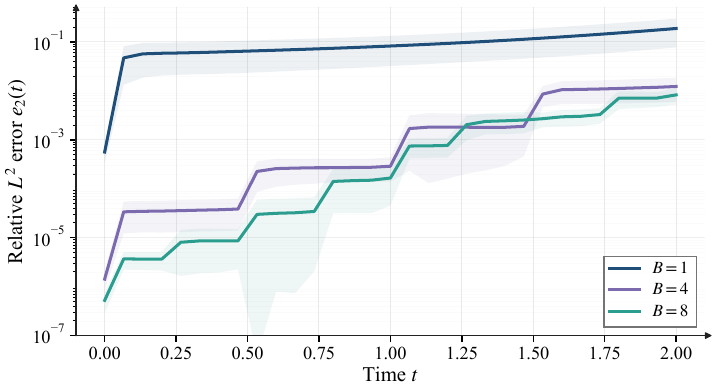}
  \captionof{figure}{Block marching controls long-time Allen--Cahn error over $T=2.0$. Curves and bands show mean $\pm$ sample standard deviation over five random seeds.}
  \label{fig:block-marching}
  \end{minipage}
\end{center}

\begin{center}
  \begin{minipage}{0.9\linewidth}
  \centering
  \maybeincludegraphics[width=\linewidth]{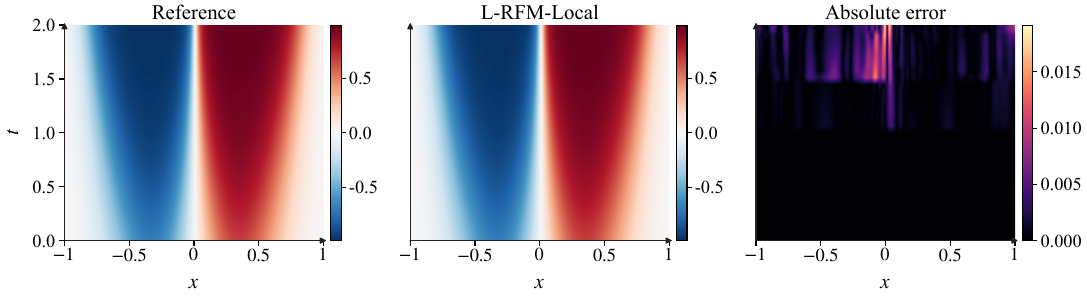}
  \captionof{figure}{Representative long-time Allen--Cahn solution from block marching.}
  \label{fig:block-marching-solution}
  \end{minipage}
\end{center}

\subsubsection{Wall-clock timing and implementation bottlenecks}
\label{subsec:timing}

The present implementation evaluates features patch by patch and recomputes the LS factorization at every Picard step.
For a row count $N$ and readout dimension $P$, dense assembly scales essentially with the number of active derivative evaluations, while the SVD-based LS solve scales like a dense overdetermined solve in $N$ and $P$ and is repeated across Picard steps for nonlinear PDEs.
The local PoU support reduces the number of active features per spatial row in principle, and this structure points to sparse assembly and vectorized patch evaluation as natural implementation improvements.
Table~\ref{tab:timing} reports a fixed-capacity Allen--Cahn timing comparison at the target-accuracy scale.
All methods use the same readout dimension $P=240$; for L-RFM-Local this corresponds to $(K,M)=(4,60)$, while the global methods use 240 scalar features.
The timing row in Table~\ref{tab:timing} was rerun on the same CPU server equipped with an Intel(R) Xeon(R) Platinum 8476C processor.
This fixed-capacity construction keeps the wall-clock comparison tied to a matched one-dimensional readout dimension.
At this fixed capacity, all methods reach the $10^{-4}$ target scale on Allen--Cahn, with the two L-RFM variants giving the smallest mean relative errors.
L-RFM-Global is competitive in wall-clock time, whereas L-RFM-Local is slower because the current implementation evaluates patchwise liquid features and recomputes dense LS factorizations during Picard iteration.
The reported wall-clock times include feature evaluation, matrix assembly, all Picard iterations, and repeated dense LS factorizations, but exclude reference-solution generation and plotting.
The table reports five-seed mean relative errors and wall-clock times with sample standard deviations, providing a practical reference point for future vectorized feature evaluation, factorization reuse, and sketched LS solvers.

\begin{center}
  \begin{minipage}{\linewidth}
  \centering
  \captionof{table}{Wall-clock cost and accuracy at fixed capacity for Allen--Cahn 1D. Error and time entries are mean $\pm$ sample standard deviation over five random-seed runs at the $10^{-4}$ relative-$L^2$ target scale; timings include all Picard LS solves but exclude reference generation and plotting.}
  \label{tab:timing}
  \TableFont
  \setlength{\tabcolsep}{3pt}
  \begin{tabular}{@{}llccc@{\hspace{8pt}}cc@{}}
    \toprule
    & & \multicolumn{3}{c}{Static baselines} & \multicolumn{2}{c}{L-RFM variants} \\
    \cmidrule(lr){3-5}\cmidrule(l){6-7}
    Problem & Quantity & ST-RFM-SoV & ST-RFM-STC & PIELM & L-RFM-Local & L-RFM-Global \\
    \midrule
    \multirow{2}{*}{Allen--Cahn 1D} & Rel.\ $L^2$ & $(2.53\pm0.72){\times}10^{-5}$ & $(4.33\pm0.85){\times}10^{-5}$ & $(4.77\pm0.85){\times}10^{-5}$ & $(3.31\pm0.82){\times}10^{-6}$ & $\boldsymbol{(1.40\pm0.40){\times}10^{-6}}$ \\
    & Time (s) & $(5.74\pm0.10){\times}10^{0}$ & $(5.24\pm0.04){\times}10^{0}$ & $\boldsymbol{(9.62\pm0.69){\times}10^{-1}}$ & $(1.85\pm0.02){\times}10^{1}$ & $(4.86\pm0.07){\times}10^{0}$ \\
    \bottomrule
  \end{tabular}
  \end{minipage}
\end{center}

\FloatBarrier

\section{Discussion and Conclusions}
\label{sec:discussion}
\label{sec:conclusion}

L-RFM addresses a central obstacle in space--time random-feature collocation for evolutionary PDEs: standard frozen activations leave temporal scales implicit, whereas stiff, dispersive, and multi-scale dynamics are often governed by distinct relaxation or propagation time scales.
The method embeds a sampled relaxation spectrum directly into closed-form liquid time-constant responses, producing temporally structured trial spaces while preserving the algebraic simplicity of a linear least-squares (LS) readout.
All hidden parameters remain frozen, residual derivatives are analytic, and nonlinear PDEs are handled through Picard-linearized LS updates.
The resulting construction is a trial-space design for mesh-free LS collocation, with its approximation role supported by the density result in Section~\ref{sec:theory} and its finite-feature behavior assessed by the numerical study in Section~\ref{sec:experiments}.

The numerical results give three main messages.
First, the heat-equation tests and tensor-product multidimensional checks verify the analytic derivative assembly, partition-of-unity construction, and Laplacian terms, while the nonlinear examples exercise the Picard-linearized residual assembly.
Second, the matched-capacity comparisons in Table~\ref{tab:main-baselines} show accuracy gains over static random-feature baselines on reaction--diffusion, nonlinear transport, dispersive, and complex-valued systems.
Third, the prescribed temporal-scale test, the Allen--Cahn diffusion-scale sweep in Table~\ref{tab:stiff} and Figure~\ref{fig:stiff}, and the fixed-$\tau$ ablations in Figure~\ref{fig:ablations} identify the relaxation-scale law as a finite-feature design variable in the tested regimes.
The conditioning measurements in Table~\ref{tab:condition-panorama} complement this conclusion by showing how the temporal feature design changes the geometry of the assembled row-weighted LS systems.

The local and global variants separate temporal representation from spatial localization.
L-RFM-Global is natural for spatially smooth or coherent responses in which temporal-scale representation is the main approximation bottleneck.
L-RFM-Local is better suited to localized interfaces, spatially nonuniform structures, and higher-derivative residuals, where PoU weights improve locality and column geometry.
The Burgers, Allen--Cahn, KdV, and NLS examples therefore compare two complementary trial-space designs within the same liquid-feature framework.

The results also point to clear extensions.
The multidimensional tests verify assembly and temporal-scale transfer in two and three spatial dimensions.
Broader high-dimensional benchmarks would further clarify the dependence on spatial dimension and geometry.
On the theoretical side, finite-$M$ PDE error, conditioning, feature-count scaling, Picard convergence, block-marching error propagation, and the full coupling through $s(\bx)$ and $\alpha(\bx)$ are natural targets for analysis beyond the qualitative density result.
On the computational side, Table~\ref{tab:timing} shows that the current implementation can be accelerated by vectorizing feature evaluation across feature--patch indices, exploiting local support with sparse or blocked linear algebra, and reusing LS factorizations during Picard iteration.
Taken together, the theory and experiments show that relaxation-scale representation can be built into frozen random features without leaving the linear LS regime, providing a high-accuracy continuous space--time approximation framework for evolutionary PDEs in which temporal scales control the finite-dimensional approximation.

\section*{Acknowledgements}
This work was carried out during Jiale Linghu's visit to the Department of Mathematics, National University of Singapore. The authors are grateful to Professor Weizhu Bao for his valuable suggestions and insightful discussions, which helped improve the development and presentation of this work.

\appendix
\renewcommand{\theHsection}{appendix.\Alph{section}}
\renewcommand{\theHsubsection}{appendix.\Alph{section}.\arabic{subsection}}
\renewcommand{\theHsubsubsection}{appendix.\Alph{section}.\arabic{subsection}.\arabic{subsubsection}}
\renewcommand{\theHequation}{appendix.\Alph{section}.\arabic{equation}}
\renewcommand{\theHfigure}{appendix.\Alph{section}.\arabic{figure}}
\renewcommand{\theHtable}{appendix.\Alph{section}.\arabic{table}}
\section{Proofs for Trial-Space and Temporal-Scale Results}
\label{app:theory-proofs}

This appendix provides the details behind the density theorem and the temporal-rank calculation in Section~\ref{sec:theory}.

\subsection{Proof of the density theorem}
\label{app:density-proof}

The proof combines a temporal density lemma with the classical density of neural ridge functions in space.

\begin{lemma}[Temporal density]
\label{lem:temporal-density}
For any $0<\alpha_0<\alpha_1<\infty$, the linear span of $\bigl\{e^{-\alpha t}:\alpha\in[\alpha_0,\alpha_1]\bigr\}$ is dense in $C([0,T])$ in the supremum norm.
\end{lemma}

In the L-RFM notation, $\alpha=1/\tau$.
Thus a compact interval $\tau\in[\tau_{\min},\tau_{\max}]$ corresponds to $\alpha\in[\tau_{\max}^{-1},\tau_{\min}^{-1}]$.

\begin{proof}
By Hahn--Banach and Riesz representation, it suffices to show that every finite signed regular Borel measure $\mu$ on $[0,T]$ annihilating this span is zero.
Suppose
\[
F(\alpha):=\int_0^T e^{-\alpha t}\,d\mu(t)=0
\qquad\text{for all }\alpha\in[\alpha_0,\alpha_1].
\]
Since $\mu$ is compactly supported, dominated convergence extends $F$ to an entire function of $\alpha\in\mathbb{C}$.
Since $F$ vanishes on a real interval with an accumulation point, the identity theorem gives $F\equiv0$ on $\mathbb{C}$.
Differentiating at any $\alpha_\star>0$ gives
\[
\int_0^T t^m e^{-\alpha_\star t}\,d\mu(t)=0,\qquad m=0,1,\ldots .
\]
Multiplication by $e^{-\alpha_\star t}$ is a bijective isomorphism of $C([0,T])$, and polynomials are dense in $C([0,T])$ by Weierstrass.
Thus $\int f\,d\mu=0$ for every $f\in C([0,T])$, so $\mu=0$.
\end{proof}

We also use the standard ridge-function density result: since $\tanh$ is continuous, bounded, and nonpolynomial, the span of $\bx\mapsto\tanh(\mathbf a^\top\bx+b)$ is dense in $C(X)$ for every compact $X\subset\R^d$ \cite{cybenko1989approximation,leshno1993multilayer}.
The nonsingular affine coordinate maps used on local patches preserve this ridge span by reparameterization.

\begin{proof}[Proof of Theorem~\ref{thm:lrfm-universal}]
Fix $u\in C(Q)$ and $\varepsilon>0$.
Let $X\subset\R^d$ be compact and let $\cR_X$ be the ridge span on $X$.
Let $\cE$ be the span of $t\mapsto e^{-t/\tau}$ for $\tau\in[\tau_{\min},\tau_{\max}]$.
By the ridge density result and Lemma~\ref{lem:temporal-density}, $\overline{\cR_X}=C(X)$ and $\overline{\cE}=C([0,T])$.
The product span $\cR_X\otimes\cE$ is therefore dense in the finite tensor-product span $C(X)\otimes C([0,T])$.
Since the latter span is an algebra that contains constants and separates points, Stone--Weierstrass gives
\[
\overline{\cR_X\otimes\cE}=C(X\times[0,T]).
\]

For the global trial space, take $X=\overline{\Omega}$ after the global affine normalization.
The tensor-product result gives uniform approximation by finite sums of ridge--exponential functions, and Assumption~\ref{a-global} places the corresponding separable feature slice \eqref{eq:subdict} in the closure of the admissible global L-RFM parameter set.
Continuity of the Duhamel feature map on compact parameter sets then transfers the approximation to $\cA_{\mathrm{glob}}$.

For the localized trial space, apply the same tensor-product result on each patch support $X_k=\mathrm{supp}\,\psi_k$ in the local affine coordinate.
Assumption~\ref{a-tau}, the separable slice \eqref{eq:subdict}, and continuity of the feature map imply that, for every $k$, there exists a finite local L-RFM expansion $p_k$ such that
\[
\sup_{(\bx,t)\in X_k\times[0,T]}\lvert u(\bx,t)-p_k(\bx,t)\rvert<\varepsilon .
\]
Define $p=\sum_{k=1}^K\psi_k p_k\in\cA_{\mathrm{loc}}$.
Using $\psi_k\ge0$ and $\sum_k\psi_k=1$,
\[
\lvert u(\bx,t)-p(\bx,t)\rvert
\le \sum_k\psi_k(\bx)\,\|u-p_k\|_\infty
\le \varepsilon .
\]
Zero-coefficient padding makes the feature count uniform across patches.
This proves uniform density of $\cA_{\mathrm{loc}}$.
Since $Q$ has finite measure, uniform convergence implies convergence in $L^p(Q)$ for every $1\le p<\infty$.
\end{proof}

\subsection{Proof of the temporal-rank proposition}
\label{app:scale-proof}

\begin{proof}[Proof of Proposition~\ref{prop:scale-identifiability}]
The functions $\{e^{-\lambda_j t}\}_{j=1}^J$ with distinct real $\lambda_j$ form a Chebyshev system on every interval.
Equivalently, any nontrivial linear combination $\sum_j c_j e^{-\lambda_j t}$ has at most $J-1$ zeros on $[0,T]$.
Since the grid has $N\ge J$ distinct points, a nonzero coefficient vector cannot vanish at all grid points, so the columns of $V$ are linearly independent.
The exact-representation statement follows by construction.
The single-rate statement is the orthogonal projection of $\mathbf y$ onto the one-dimensional span of $\mathbf v_{\lambda_0}$.
\end{proof}

The full-rank statement is an algebraic temporal calculation.
Clustered rates, an unresolved time grid, or a PDE residual that mixes temporal and spatial derivatives can still produce an ill-conditioned collocation matrix.

\subsection{Scope of the theoretical results}
\label{app:theory-remarks}

The density theorem has a direct infinite-sample support consequence: if the sampling-law supports have closures containing the deterministic parameter slices used above, then an infinite i.i.d.\ draw is almost surely dense in the same trial spaces by a standard Borel--Cantelli argument.
This support statement is distinct from finite-feature probability bounds and residual-error estimates.
Similarly, the frozen feature family can be interpreted as inducing a space--time random-feature kernel, whereas the analysis here uses deterministic approximation and temporal-rank arguments rather than RKHS, Monte Carlo, or kernel-regression bounds.
The residual-collocation minimizer depends on the PDE operator, row sampling, row weights, SVD truncation, analytic residual derivatives, and, for nonlinear problems, the Picard iterate.
The conditioning results in Section~\ref{sec:experiments} are empirical measurements for the assembled weighted LS systems and complement, rather than follow directly from, Proposition~\ref{prop:scale-identifiability}.

\section{Experimental Protocol and Supplementary Results}
\label{app:configurations}

This appendix records the numerical protocol and supplementary results supporting Section~\ref{sec:experiments}.
The first subsection states the common solver, weighting, repeat, and metric rules.
The second gives the reference solutions, retained problem parameters, partition geometry, and collocation row counts.
The final subsection reports extended $L^\infty$ errors and empirical conditioning results for the benchmark rows most relevant to temporal representation and local/global support.

\subsection{Common numerical protocol}
\label{app:implementation}

All closed-form features and their analytic derivatives are evaluated in double precision.
Unless otherwise stated, the least-squares solve uses the row-weighted system \eqref{eq:lin-system} and a singular-value decomposition with relative truncation threshold $\sigma_{\mathrm{cut}}=10^{-12}$.
The block weights follow the normalization rule in Section~\ref{subsec:solver}:
\[
\lambda_{\mathrm{PDE}}=1,\qquad
\lambda_b=\left(\frac{N_{\mathrm{PDE}}}{N_b}\right)^{1/2},
\quad b\in\mathcal{B}_{\mathrm{aux}},
\]
where $\mathcal{B}_{\mathrm{aux}}$ contains the nonempty initial, boundary, periodic, or interface-matching residual blocks.
With this convention, each residual block contributes on the scale of its mean squared residual rather than on the scale of its raw row count.
These weights, the SVD threshold, the collocation sets, the evaluation grid, and the Picard settings are held fixed across methods on the same benchmark row.
For nonlinear rows, Picard iteration stops when the configured relative update tolerance is reached or when the maximum iteration count is reached; such runs report the last Picard iterate.
For block-marched rows, the reported collocation row count is per block and a fresh feature family is drawn on each block.

Random feature parameters and collocation points are drawn from repeated random streams recorded together with the configuration.
All numerical experiments in this paper are run with five independent random seeds; quantitative entries report the mean $\pm$ sample standard deviation over those five runs.
Field-visualization figures display one fixed seed from the same five-seed batch for qualitative inspection, while the accompanying numerical summaries use the five-seed statistics.
Timing rows use the same five seeded runs and report mean $\pm$ sample standard deviation.
Condition numbers are computed from the assembled row-weighted LS matrix used for the readout solve, and effective ranks use the same relative SVD threshold as the solver.

\subsection{Reference solutions and retained configurations}
\label{app:configs}

Reference solutions are fixed before each comparison.
Heat 1D, KdV, and focusing NLS use closed-form references; Burgers uses the Cole--Hopf transformation \cite{cole1951quasilinear,hopf1950burgers} evaluated by high-accuracy quadrature.
The NLS reference used in Section~\ref{subsec:local-global-1d} is
\begin{equation}
\begin{aligned}
\psi_\star(x,t)
&=\eta\,\mathrm{sech}\bigl(\eta(x{-}x_0{-}vt)\bigr)
  \exp\!\bigl(i[v(x{-}x_0)+\tfrac12(\eta^2{-}v^2)t]\bigr),\\
&\hspace{1.2em}\eta=1,\qquad v=0.5,\qquad x_0=-2 .
\end{aligned}
\label{eq:nls-exact}
\end{equation}
Allen--Cahn uses a Fourier pseudospectral reference \cite{trefethen2000spectral}.
The 2D heat and 3D prescribed temporal-scale references are stated in Sections~\ref{subsec:verification-convergence} and \ref{subsec:scope-checks}.

Table~\ref{tab:appendix-configs} lists the retained domains, spatial layouts, feature counts, and row counts, grouped according to the main-text experiment sequence.
On each benchmark row, all methods share the same collocation sets, row weights, SVD threshold, nonlinear iteration settings, and reference/evaluation grid.
Fixed-capacity comparisons match the total scalar readout dimension $P=KM$; for multidimensional rows, $K$ denotes the tensor-product PoU layout, and global methods use the corresponding total $P$.

\begin{center}
\centering
  \captionof{table}{Experimental configurations used for the main-text and supplementary results.}
\label{tab:appendix-configs}
\TableFont
\setlength{\tabcolsep}{2pt}
\begin{tabular}{@{}>{\raggedright\arraybackslash}p{0.16\linewidth}>{\raggedright\arraybackslash}p{0.33\linewidth}>{\raggedright\arraybackslash}p{0.12\linewidth}>{\raggedright\arraybackslash}p{0.19\linewidth}>{\raggedright\arraybackslash}p{0.14\linewidth}@{}}
\toprule
PDE & Domain $\times[0,T]$ & Layout & Features & Rows/protocol \\
\midrule
\multicolumn{5}{@{}l}{\emph{Verification and temporal representation}} \\
Heat 1D conv.\ & $[-1,1]\times[0,0.25]$ & $\{1,2,4,8\}$ & $\{25,50,100\}$ & $(900,220,120)$ \\
Heat 1D colloc.\ & $[-1,1]\times[0,0.25]$ & 4 & 100 & varied; see Table~\ref{tab:collocation-sensitivity} \\
Heat 1D verification & $[-1,1]\times[0,0.25]$ & 4 & 100 & $(900,220,120)$ \\
Temporal-scale test & $[-1,1]\times[0,0.1]$, $\tau=\{1,10^{-2},10^{-4}\}$ & $K=4$/global & $P\in\{40,80,160,320,640\}$ & $(1600,320,140)$ \\
\addlinespace[2pt]
\multicolumn{5}{@{}l}{\emph{Mechanism and controlled ablations}} \\
Stiff AC sweep & $[-1,1]\times[0,0.1]$ & 4 & 100 & $(2000,400,120)$ \\
Static ablation & heat/AC/Burgers & as source row & as source row & same rows as source row \\
$\tau$ ablation & AC 1D & 4 & 100 & same rows as Allen--Cahn \\
\addlinespace[2pt]
\multicolumn{5}{@{}l}{\emph{Representative one-dimensional PDE comparison}} \\
Allen--Cahn 1D & $[-1,1]\times[0,0.1]$, $\epsilon{=}10^{-4}$ & 4 & 100 & $(2000,400,120)$ \\
Burgers 1D & $[-1,1]\times[0,0.10]$, $\nu{=}0.005/\pi$ & 6 & 80 & $(550,180,80)$ \\
KdV 1D & $[-8,8]\times[0,0.5]$, $c{=}2,x_0{=}-3$ & 6 & 80 & $(900,240,120)$ \\
NLS 1D & $[-8,8]\times[0,0.4]$, $\eta{=}1,v{=}0.5$ & 6 & 80 & $(3000,400,160)$ \\
\addlinespace[2pt]
\multicolumn{5}{@{}l}{\emph{Multidimensional known-solution checks}} \\
Heat 2D & $[0,1]^2\times[0,0.12]$ & $2\times 2$ & $\{20,30,40,60\}$ & 1370 total rows \\
Temporal-scale 3D & $[-1,1]^3\times[0,0.1]$, $\tau=\{1,10^{-2},10^{-4}\}$ & $K=8$/global & 480 local or $P=3840$ global & $(5200,1000,160)$ \\
\addlinespace[2pt]
\multicolumn{5}{@{}l}{\emph{Robustness, conditioning, and cost}} \\
AC long-time & $[-1,1]\times[0,2.0]$, $B\in\{1,4,8\}$ & 4 & 80 & per block as Allen--Cahn row \\
Timing & Allen--Cahn 1D & $K=4$/global & $P=240$; local $M=60$ & target-row protocol \\
\bottomrule
\end{tabular}
\end{center}

\subsection{Supplementary error and empirical conditioning results}
\label{app:baselines}

This subsection reports supplementary results for the one-dimensional rows that most directly test temporal representation and local/global support: Heat 1D, Allen--Cahn 1D, Burgers 1D, KdV 1D, and NLS 1D.
It also includes the stiff Allen--Cahn sweep used to examine conditioning as $\epsilon$ decreases.
The multidimensional known-solution checks are described in Section~\ref{subsec:scope-checks} and configured in Table~\ref{tab:appendix-configs}.

Table~\ref{tab:appendix-baselines} adds two result blocks to the selected one-dimensional subset: five-seed relative $L^\infty$ errors and empirical LS condition numbers.
These results support the main-text trends with a stronger pointwise error norm and with the singular spectrum of the assembled LS matrix.
The $L^\infty$ block broadly follows the five-seed mean-error picture, with Burgers 1D favoring L-RFM-Global and KdV 1D favoring L-RFM-Local.
The condition-number block is empirical and is computed from the row-weighted LS matrix used in the readout solve.
In this representative subset, L-RFM-Local has the smallest singular spectrum ratio across every nontrivial row, sometimes by four to five orders of magnitude.
For example, on NLS at matched real-block dimension $P=960$, L-RFM-Local has five-seed mean $\kappa\approx 1.0{\times}10^{11}$, while ST-RFM and L-RFM-Global cluster around $10^{15}$ to $10^{16}$.
This behavior is consistent with the multi-scale $\tau$ law supplying useful temporal directions and local support reducing spatial column alignment, but it remains empirical rather than a theorem.
On these five-seed rows, lower errors are often accompanied by smaller empirical condition numbers; the observation is used as numerical evidence about the assembled systems, not as a conditioning theorem.

\begin{center}
\centering
\captionof{table}{Supplementary error and conditioning results for selected one-dimensional benchmarks. Relative $L^\infty$ errors and empirical LS condition numbers are reported as mean $\pm$ sample standard deviation over five random seeds.}
\label{tab:appendix-baselines}
\TableFont
\setlength{\tabcolsep}{4pt}
\begin{tabular}{@{}llccc@{\hspace{8pt}}cc@{}}
\toprule
& & \multicolumn{3}{c}{Static baselines} & \multicolumn{2}{c}{L-RFM variants} \\
\cmidrule(lr){3-5}\cmidrule(l){6-7}
Metric & Problem & ST-RFM-SoV & ST-RFM-STC & PIELM & L-RFM-Local & L-RFM-Global \\
\midrule
\multirow{5}{*}{Rel.\ $L^\infty$}
 & Heat 1D        & $(6.83\pm1.35){\times}10^{-7}$ & $(1.93\pm0.84){\times}10^{-7}$ & $\boldsymbol{(1.39\pm0.81){\times}10^{-8}}$ & $(2.45\pm3.16){\times}10^{-6}$ & $(9.18\pm10.00){\times}10^{-8}$ \\
 & Allen--Cahn 1D & $(1.14\pm0.30){\times}10^{-5}$ & $(1.93\pm0.87){\times}10^{-5}$ & $(8.41\pm4.65){\times}10^{-6}$ & $\boldsymbol{(1.19\pm0.29){\times}10^{-7}}$ & $(8.61\pm2.73){\times}10^{-7}$ \\
 & Burgers 1D     & $(2.90\pm0.95){\times}10^{-5}$ & $(2.36\pm1.02){\times}10^{-5}$ & $(3.62\pm1.33){\times}10^{-5}$ & $(1.04\pm0.48){\times}10^{-4}$ & $\boldsymbol{(5.65\pm1.86){\times}10^{-6}}$ \\
 & KdV 1D         & $(1.48\pm1.48){\times}10^{-2}$ & $(1.38\pm0.70){\times}10^{-2}$ & $(2.22\pm1.57){\times}10^{-2}$ & $\boldsymbol{(2.56\pm0.96){\times}10^{-3}}$ & $(5.06\pm6.52){\times}10^{-3}$ \\
 & NLS 1D         & $(2.16\pm0.55){\times}10^{-3}$ & $(4.81\pm0.91){\times}10^{-3}$ & $(2.59\pm2.86){\times}10^{-2}$ & $\boldsymbol{(1.33\pm0.46){\times}10^{-4}}$ & $(3.03\pm1.73){\times}10^{-4}$ \\
\midrule
\multirow{5}{*}{$\kappa(\mathbf{A})$}
 & Heat 1D        & $(4.41\pm1.70){\times}10^{16}$ & $(1.31\pm0.59){\times}10^{16}$ & $(6.49\pm2.01){\times}10^{16}$ & $\boldsymbol{(5.55\pm3.41){\times}10^{13}}$ & $(5.60\pm1.46){\times}10^{16}$ \\
 & Allen--Cahn 1D & $(6.13\pm1.17){\times}10^{16}$ & $(3.03\pm0.47){\times}10^{16}$ & $(5.59\pm5.33){\times}10^{14}$ & $\boldsymbol{(2.75\pm1.90){\times}10^{13}}$ & $(2.88\pm1.04){\times}10^{16}$ \\
 & Burgers 1D     & $(2.74\pm0.96){\times}10^{16}$ & $(5.42\pm1.80){\times}10^{15}$ & $(1.77\pm0.96){\times}10^{16}$ & $\boldsymbol{(3.69\pm4.15){\times}10^{12}}$ & $(9.15\pm2.15){\times}10^{16}$ \\
 & KdV 1D         & $(9.55\pm7.67){\times}10^{15}$ & $(6.23\pm3.16){\times}10^{14}$ & $(6.35\pm4.05){\times}10^{15}$ & $\boldsymbol{(2.41\pm5.38){\times}10^{12}}$ & $(2.40\pm5.28){\times}10^{16}$ \\
 & NLS 1D         & $(2.53\pm1.47){\times}10^{16}$ & $(2.65\pm1.09){\times}10^{15}$ & $(1.33\pm0.90){\times}10^{16}$ & $\boldsymbol{(1.02\pm0.38){\times}10^{11}}$ & $(1.77\pm0.62){\times}10^{16}$ \\
\bottomrule
\end{tabular}
\end{center}

Table~\ref{tab:appendix-conditioning} extends this view to the stiff Allen--Cahn sweep: as $\epsilon$ decreases from $10^{-1}$ to $10^{-4}$, the five-seed mean L-RFM-Local condition number stays within roughly one order of magnitude of $10^{13}$, while the static localized baselines remain between $10^{16}$ and $10^{17}$ and L-RFM-Global stays between $10^{15}$ and $10^{16}$.
The pattern across rows is informative: the L-RFM-Local matrix remains in the best-conditioned group throughout the sweep and has the clearest advantage at the stiffest setting.
This behavior is consistent with the multi-scale $\tau$ law providing useful temporal directions and the local support reducing spatial column alignment, but it remains an empirical conditioning result rather than a theorem.

\begin{center}
\centering
\captionof{table}{Supplementary conditioning sweep for stiff Allen--Cahn least-squares systems. Values are reported as mean $\pm$ sample standard deviation over five random seeds.}
\label{tab:appendix-conditioning}
\TableFont
\begin{tabular}{@{}lccc@{\hspace{8pt}}cc@{}}
\toprule
& \multicolumn{3}{c}{Static baselines} & \multicolumn{2}{c}{L-RFM variants} \\
\cmidrule(lr){2-4}\cmidrule(l){5-6}
$\epsilon$ & ST-RFM-SoV & ST-RFM-STC & PIELM & L-RFM-Local & L-RFM-Global \\
\midrule
$10^{-1}$ & $(4.26\pm0.71){\times}10^{16}$ & $(3.30\pm0.87){\times}10^{16}$ & $(4.13\pm1.49){\times}10^{14}$ & $\boldsymbol{(5.16\pm2.16){\times}10^{12}}$ & $(1.27\pm0.38){\times}10^{15}$ \\
$10^{-2}$ & $(4.87\pm0.69){\times}10^{16}$ & $(3.47\pm0.86){\times}10^{16}$ & $(1.00\pm1.09){\times}10^{15}$ & $\boldsymbol{(1.99\pm0.72){\times}10^{13}}$ & $(6.23\pm1.38){\times}10^{15}$ \\
$10^{-3}$ & $(5.58\pm0.90){\times}10^{16}$ & $(3.23\pm0.41){\times}10^{16}$ & $(4.69\pm2.70){\times}10^{14}$ & $\boldsymbol{(3.53\pm1.72){\times}10^{13}}$ & $(1.00\pm0.45){\times}10^{16}$ \\
$10^{-4}$ & $(5.35\pm1.03){\times}10^{16}$ & $(2.66\pm0.35){\times}10^{16}$ & $(2.59\pm5.11){\times}10^{15}$ & $\boldsymbol{(4.70\pm3.01){\times}10^{13}}$ & $(1.69\pm0.38){\times}10^{16}$ \\
\bottomrule
\end{tabular}
\end{center}

\def\bibsection{\section*{References}}
\bibliographystyle{elsarticle-num}
\bibliography{refs}

@article{chen2023strfm,
  title = {The Random Feature Method for Time-Dependent Problems},
  author = {Chen, Jing-Run and E, Weinan and Luo, Yi-Xin},
  journal = {East Asian Journal on Applied Mathematics},
  volume = {13},
  number = {3},
  pages = {435--463},
  year = {2023},
  doi = {10.4208/eajam.2023-065.050423}
}

@book{leveque2007finite,
  title = {Finite Difference Methods for Ordinary and Partial Differential Equations: Steady-State and Time-Dependent Problems},
  author = {LeVeque, Randall J.},
  publisher = {Society for Industrial and Applied Mathematics},
  address = {Philadelphia},
  year = {2007},
  doi = {10.1137/1.9780898717839},
  isbn = {978-0-89871-629-0}
}

@article{hughes1988spacetime,
  title = {Space-Time Finite Element Methods for Elastodynamics: Formulations and Error Estimates},
  author = {Hughes, Thomas J. R. and Hulbert, Gregory M.},
  journal = {Computer Methods in Applied Mechanics and Engineering},
  volume = {66},
  number = {3},
  pages = {339--363},
  year = {1988},
  doi = {10.1016/0045-7825(88)90006-0}
}

@article{chen2022rfm,
  title = {Bridging Traditional and Machine Learning-Based Algorithms for Solving {PDE}s: The Random Feature Method},
  author = {Chen, Jingrun and Chi, Xurong and E, Weinan and Yang, Zhouwang},
  journal = {Journal of Machine Learning},
  volume = {1},
  number = {3},
  pages = {268--298},
  year = {2022},
  doi = {10.4208/jml.220726}
}

@article{chi2024interface,
  title = {The Random Feature Method for Solving Interface Problems},
  author = {Chi, Xurong and Chen, Jingrun and Yang, Zhouwang},
  journal = {Computer Methods in Applied Mechanics and Engineering},
  volume = {420},
  pages = {116719},
  year = {2024},
  doi = {10.1016/j.cma.2023.116719}
}

@article{sun2025twolevel,
  title = {Two-Level Random Feature Methods for Elliptic Partial Differential Equations over Complex Domains},
  author = {Sun, Yifei and Chen, Jingrun},
  journal = {Computer Methods in Applied Mechanics and Engineering},
  volume = {441},
  pages = {117961},
  year = {2025},
  doi = {10.1016/j.cma.2025.117961}
}

@article{song2026discontinuity,
  title = {Discontinuity-Capturing Random Feature Method for Interface Problems},
  author = {Song, Wentian and Chi, Xurong and Yang, Zhouwang and Cheng, Wan and Chen, Jingrun},
  journal = {Computer Methods in Applied Mechanics and Engineering},
  volume = {453},
  pages = {118841},
  year = {2026},
  doi = {10.1016/j.cma.2026.118841}
}

@inproceedings{rahimi2007random,
  title = {Random Features for Large-Scale Kernel Machines},
  author = {Rahimi, Ali and Recht, Benjamin},
  booktitle = {Advances in Neural Information Processing Systems},
  volume = {20},
  pages = {1177--1184},
  year = {2007}
}

@article{raissi2019pinn,
  title = {Physics-Informed Neural Networks: A Deep Learning Framework for Solving Forward and Inverse Problems Involving Nonlinear Partial Differential Equations},
  author = {Raissi, Maziar and Perdikaris, Paris and Karniadakis, George Em},
  journal = {Journal of Computational Physics},
  volume = {378},
  pages = {686--707},
  year = {2019},
  doi = {10.1016/j.jcp.2018.10.045}
}

@article{dwivedi2020pielm,
  title = {Physics Informed Extreme Learning Machine ({PIELM})--A Rapid Method for the Numerical Solution of Partial Differential Equations},
  author = {Dwivedi, Vikas and Srinivasan, Balaji},
  journal = {Neurocomputing},
  volume = {391},
  pages = {96--118},
  year = {2020},
  doi = {10.1016/j.neucom.2019.12.099}
}

@article{dong2021localelm,
  title = {Local Extreme Learning Machines and Domain Decomposition for Solving Linear and Nonlinear Partial Differential Equations},
  author = {Dong, Suchuan and Li, Zongwei},
  journal = {Computer Methods in Applied Mechanics and Engineering},
  volume = {387},
  pages = {114129},
  year = {2021},
  doi = {10.1016/j.cma.2021.114129}
}

@article{dong2022hyperparameter,
  title = {On Computing the Hyperparameter of Extreme Learning Machines: Algorithm and Application to Computational {PDE}s, and Comparison with Classical and High-Order Finite Elements},
  author = {Dong, Suchuan and Yang, Jielin},
  journal = {Journal of Computational Physics},
  volume = {463},
  pages = {111290},
  year = {2022},
  doi = {10.1016/j.jcp.2022.111290}
}

@article{ni2023hcelm,
  title = {Numerical Computation of Partial Differential Equations by Hidden-Layer Concatenated Extreme Learning Machine},
  author = {Ni, Naxian and Dong, Suchuan},
  journal = {Journal of Scientific Computing},
  volume = {95},
  number = {2},
  pages = {35},
  year = {2023},
  doi = {10.1007/s10915-023-02162-0}
}

@article{wang2024elmhighdim,
  title = {An Extreme Learning Machine-Based Method for Computational {PDE}s in Higher Dimensions},
  author = {Wang, Yiran and Dong, Suchuan},
  journal = {Computer Methods in Applied Mechanics and Engineering},
  volume = {418},
  pages = {116578},
  year = {2024},
  doi = {10.1016/j.cma.2023.116578}
}

@article{hasani2021ltc,
  title = {Liquid Time-Constant Networks},
  author = {Hasani, Ramin and Lechner, Mathias and Amini, Alexander and Rus, Daniela and Grosu, Radu},
  journal = {Proceedings of the AAAI Conference on Artificial Intelligence},
  volume = {35},
  number = {9},
  pages = {7657--7666},
  year = {2021},
  doi = {10.1609/aaai.v35i9.16936}
}

@article{sun2024piln,
  title = {A New Method for Solving Nonlinear Partial Differential Equations Based on Liquid Time-Constant Networks},
  author = {Sun, Jiuyun and Dong, Huanhe and Fang, Yong},
  journal = {Journal of Systems Science and Complexity},
  volume = {37},
  number = {2},
  pages = {480--493},
  year = {2024},
  doi = {10.1007/s11424-024-3349-z}
}

@article{schiassi2021xtfc,
  title = {Extreme Theory of Functional Connections: A Fast Physics-Informed Neural Network Method for Solving Ordinary and Partial Differential Equations},
  author = {Schiassi, Enrico and De Florio, Mario and D'Ambrosio, Andrea and Mortari, Daniele and Furfaro, Roberto},
  journal = {Neurocomputing},
  volume = {457},
  pages = {334--356},
  year = {2021},
  doi = {10.1016/j.neucom.2021.06.015}
}

@misc{piatti2026rcde,
  title = {Random Controlled Differential Equations},
  author = {Piatti, Francesco and Cass, Thomas and Turner, William F.},
  eprint = {2512.23670},
  archivePrefix = {arXiv},
  primaryClass = {cs.LG},
  year = {2025},
  note = {Accepted at ICLR 2026}
}

@article{cybenko1989approximation,
  title = {Approximation by Superpositions of a Sigmoidal Function},
  author = {Cybenko, George},
  journal = {Mathematics of Control, Signals and Systems},
  volume = {2},
  number = {4},
  pages = {303--314},
  year = {1989},
  doi = {10.1007/BF02551274}
}

@article{leshno1993multilayer,
  title = {Multilayer Feedforward Networks with a Nonpolynomial Activation Function Can Approximate Any Function},
  author = {Leshno, Moshe and Lin, Vladimir Ya. and Pinkus, Allan and Schocken, Shimon},
  journal = {Neural Networks},
  volume = {6},
  number = {6},
  pages = {861--867},
  year = {1993},
  doi = {10.1016/S0893-6080(05)80131-5}
}

@article{jaeger2001esn,
  title = {The ``Echo State'' Approach to Analysing and Training Recurrent Neural Networks},
  author = {Jaeger, Herbert},
  journal = {GMD Technical Report 148, German National Research Center for Information Technology},
  year = {2001}
}

@article{maass2002lsm,
  title = {Real-Time Computing without Stable States: A New Framework for Neural Computation Based on Perturbations},
  author = {Maass, Wolfgang and Natschl\"ager, Thomas and Markram, Henry},
  journal = {Neural Computation},
  volume = {14},
  number = {11},
  pages = {2531--2560},
  year = {2002},
  doi = {10.1162/089976602760407955}
}

@article{wang2022ntk,
  title = {When and Why {PINN}s Fail to Train: A Neural Tangent Kernel Perspective},
  author = {Wang, Sifan and Yu, Xinling and Perdikaris, Paris},
  journal = {Journal of Computational Physics},
  volume = {449},
  pages = {110768},
  year = {2022},
  doi = {10.1016/j.jcp.2021.110768}
}

@inproceedings{racca2021apiesn,
  title = {Automatic-Differentiated Physics-Informed Echo State Network ({API-ESN})},
  author = {Racca, Alberto and Magri, Luca},
  booktitle = {International Conference on Computational Science (ICCS)},
  series = {Lecture Notes in Computer Science},
  volume = {12746},
  pages = {323--329},
  year = {2021},
  doi = {10.1007/978-3-030-77977-1\_25}
}

@article{karniadakis2021review,
  title = {Physics-informed machine learning},
  author = {Karniadakis, George Em and Kevrekidis, Ioannis G. and Lu, Lu and Perdikaris, Paris and Wang, Sifan and Yang, Liu},
  journal = {Nature Reviews Physics},
  volume = {3},
  number = {6},
  pages = {422--440},
  year = {2021},
  doi = {10.1038/s42254-021-00314-5}
}

@article{lu2021deeponet,
  title = {Learning Nonlinear Operators via {DeepONet} Based on the Universal Approximation Theorem of Operators},
  author = {Lu, Lu and Jin, Pengzhan and Pang, Guofei and Zhang, Zhongqiang and Karniadakis, George Em},
  journal = {Nature Machine Intelligence},
  volume = {3},
  number = {3},
  pages = {218--229},
  year = {2021},
  doi = {10.1038/s42256-021-00302-5}
}

@inproceedings{li2021fno,
  title = {Fourier Neural Operator for Parametric Partial Differential Equations},
  author = {Li, Zongyi and Kovachki, Nikola and Azizzadenesheli, Kamyar and Liu, Burigede and Bhattacharya, Kaushik and Stuart, Andrew and Anandkumar, Anima},
  booktitle = {International Conference on Learning Representations},
  pages = {1--16},
  year = {2021},
  eprint = {2010.08895},
  archivePrefix = {arXiv}
}

@article{zhang2024d2no,
  title = {{D2NO}: Efficient Handling of Heterogeneous Input Function Spaces with Distributed Deep Neural Operators},
  author = {Zhang, Zecheng and Moya, Christian and Lu, Lu and Lin, Guang and Schaeffer, Hayden},
  journal = {Computer Methods in Applied Mechanics and Engineering},
  volume = {428},
  pages = {117084},
  year = {2024},
  doi = {10.1016/j.cma.2024.117084}
}

@article{huang2022hompinns,
  title = {{HomPINNs}: Homotopy Physics-Informed Neural Networks for Learning Multiple Solutions of Nonlinear Elliptic Differential Equations},
  author = {Huang, Yao and Hao, Wenrui and Lin, Guang},
  journal = {Computers \& Mathematics with Applications},
  volume = {121},
  pages = {62--73},
  year = {2022},
  doi = {10.1016/j.camwa.2022.07.002}
}

@article{zheng2024hompinns,
  title = {{HomPINNs}: Homotopy Physics-Informed Neural Networks for Solving the Inverse Problems of Nonlinear Differential Equations with Multiple Solutions},
  author = {Zheng, Haoyang and Huang, Yao and Huang, Ziyang and Hao, Wenrui and Lin, Guang},
  journal = {Journal of Computational Physics},
  volume = {500},
  pages = {112751},
  year = {2024},
  doi = {10.1016/j.jcp.2023.112751}
}

@article{huang2006elm,
  title = {Extreme Learning Machine: Theory and Applications},
  author = {Huang, Guang-Bin and Zhu, Qin-Yu and Siew, Chee-Kheong},
  journal = {Neurocomputing},
  volume = {70},
  number = {1--3},
  pages = {489--501},
  year = {2006},
  doi = {10.1016/j.neucom.2005.12.126}
}

@article{cole1951quasilinear,
  title = {On a Quasi-Linear Parabolic Equation Occurring in Aerodynamics},
  author = {Cole, Julian D.},
  journal = {Quarterly of Applied Mathematics},
  volume = {9},
  number = {3},
  pages = {225--236},
  year = {1951},
  doi = {10.1090/qam/42889}
}

@article{hopf1950burgers,
  title = {The Partial Differential Equation $u_t + u u_x = \mu u_{xx}$},
  author = {Hopf, Eberhard},
  journal = {Communications on Pure and Applied Mathematics},
  volume = {3},
  number = {3},
  pages = {201--230},
  year = {1950},
  doi = {10.1002/cpa.3160030302}
}

@article{allen1979microscopic,
  title = {A Microscopic Theory for Antiphase Boundary Motion and Its Application to Antiphase Domain Coarsening},
  author = {Allen, Samuel M. and Cahn, John W.},
  journal = {Acta Metallurgica},
  volume = {27},
  number = {6},
  pages = {1085--1095},
  year = {1979},
  doi = {10.1016/0001-6160(79)90196-2}
}

@book{trefethen2000spectral,
  title = {Spectral Methods in {MATLAB}},
  author = {Trefethen, Lloyd N.},
  publisher = {Society for Industrial and Applied Mathematics},
  year = {2000},
  doi = {10.1137/1.9780898719598},
  isbn = {978-0-89871-465-4}
}

@article{bao2016diracMTI,
  title = {A Uniformly Accurate Multiscale Time Integrator Pseudospectral Method for the {Dirac} Equation in the Nonrelativistic Limit Regime},
  author = {Bao, Weizhu and Cai, Yongyong and Jia, Xiaowei and Tang, Qinglin},
  journal = {SIAM Journal on Numerical Analysis},
  volume = {54},
  number = {3},
  pages = {1785--1812},
  year = {2016},
  doi = {10.1137/15M1032375}
}

@article{bao2016dipolarNUFFT,
  title = {Accurate and Efficient Numerical Methods for Computing Ground States and Dynamics of Dipolar {Bose--Einstein} Condensates via the Nonuniform {FFT}},
  author = {Bao, Weizhu and Tang, Qinglin and Zhang, Yong},
  journal = {Communications in Computational Physics},
  volume = {19},
  number = {5},
  pages = {1141--1166},
  year = {2016},
  doi = {10.4208/cicp.scpde14.37s}
}

@article{bao2017dewettingPFEM,
  title = {A Parametric Finite Element Method for Solid-State Dewetting Problems with Anisotropic Surface Energies},
  author = {Bao, Weizhu and Jiang, Wei and Wang, Yan and Zhao, Quan},
  journal = {Journal of Computational Physics},
  volume = {330},
  pages = {380--400},
  year = {2017},
  doi = {10.1016/j.jcp.2016.11.015}
}

@article{bao2019logSchrodinger,
  title = {Regularized Numerical Methods for the Logarithmic {Schr{\"o}dinger} Equation},
  author = {Bao, Weizhu and Carles, R{\'e}mi and Su, Chunmei and Tang, Qinglin},
  journal = {Numerische Mathematik},
  volume = {143},
  number = {2},
  pages = {461--487},
  year = {2019},
  doi = {10.1007/s00211-019-01058-2}
}

@article{lukosevicius2009reservoir,
  title = {Reservoir Computing Approaches to Recurrent Neural Network Training},
  author = {Luko{\v{s}}evi{\v{c}}ius, Mantas and Jaeger, Herbert},
  journal = {Computer Science Review},
  volume = {3},
  number = {3},
  pages = {127--149},
  year = {2009},
  doi = {10.1016/j.cosrev.2009.03.005}
}

@article{wu2023comprehensive,
  title = {A Comprehensive Study of Non-Adaptive and Residual-Based Adaptive Sampling for Physics-Informed Neural Networks},
  author = {Wu, Chenxi and Zhu, Min and Tan, Qinyang and Kartha, Yadhu and Lu, Lu},
  journal = {Computer Methods in Applied Mechanics and Engineering},
  volume = {403},
  pages = {115671},
  year = {2023},
  doi = {10.1016/j.cma.2022.115671}
}

@article{dolean2024multilevel,
  title = {Multilevel Domain Decomposition-Based Architectures for Physics-Informed Neural Networks},
  author = {Dolean, Victorita and Heinlein, Alexander and Mishra, Siddhartha and Moseley, Ben},
  journal = {Computer Methods in Applied Mechanics and Engineering},
  volume = {429},
  pages = {117116},
  year = {2024},
  doi = {10.1016/j.cma.2024.117116}
}

@article{vanbeek2026featurefilter,
  title = {Local Feature Filtering for Scalable and Well-Conditioned Domain-Decomposed Random Feature Methods},
  author = {van Beek, Jan Willem and Dolean, Victorita and Moseley, Ben},
  journal = {Computer Methods in Applied Mechanics and Engineering},
  volume = {449},
  pages = {118583},
  year = {2026},
  doi = {10.1016/j.cma.2025.118583}
}

@article{anderson2026elmfbpinn,
  title = {{ELM-FBPINNs}: An Efficient Multilevel Random Feature Method},
  author = {Anderson, Samuel and Dolean, Victorita and Moseley, Ben and Pestana, Jennifer},
  journal = {Machine Learning for Computational Science and Engineering},
  volume = {2},
  number = {1},
  pages = {23},
  year = {2026},
  doi = {10.1007/s44379-026-00071-1}
}

@article{wu2026hybrid,
  title = {Are Deep Learning Based Hybrid {PDE} Solvers Reliable? Why Training Paradigms and Update Strategies Matter},
  author = {Wu, Yuhan and van Beek, Jan Willem and Dolean, Victorita and Heinlein, Alexander},
  journal = {Computing in Science \& Engineering},
  volume = {Early Access},
  pages = {1--12},
  year = {2026},
  doi = {10.1109/MCSE.2026.3696587}
}

\end{document}